\documentclass{amsart}

\usepackage{amssymb,amsfonts,amsmath}
\usepackage{graphicx}
\usepackage[usenames,dvipsnames]{xcolor}
\usepackage{amstext}
\usepackage{amsbsy}
\usepackage{amsopn}
\usepackage{amsthm}
\usepackage{wrapfig}
\usepackage[vflt]{floatflt}
\usepackage[ruled,vlined]{algorithm2e}

\newtheorem{thm}{Theorem}[section]
\newtheorem{cor}[thm]{Corollary}
\newtheorem{lem}[thm]{Lemma}
\newtheorem{conj}[thm]{Conjecture}
\newtheorem{prop}[thm]{Proposition}
\theoremstyle{definition}

\theoremstyle{remark}

\numberwithin{equation}{section}

\newcommand{\abs}[1]{\left\vert#1\right\vert}

\newcommand{\set}[1]{\left\{#1\right\}}
\newcommand{\parr}[1]{\left (#1\right )}
\newcommand{\brac}[1]{\left [#1\right ]}

\newcommand{\Real}{\mathbb R}
\newcommand{\Natural}{\mathbb N}
\newcommand{\eps}{\varepsilon}

\newcommand{\too}{\rightarrow}

\newcommand{\A}{\mathcal{A}}

\newcommand{\bbar}[1]{\overline{#1}}
\newcommand{\wt}[1]{\widetilde{#1}} %wide tilde
\newcommand{\wh}[1]{\widehat{#1}} %wide hat

\def \i{\textbf{\footnotesize{i}}\hspace{0.5mm}} % traditional \i , which is a letter i without a dot}
\def \D{\mathcal{D}} %unit disk
\def \T{\mathcal{T}} %Equilateral [e^i {0,2pi/3,4pi/3}]
\def \A{\mathcal{A}} %affine map
\def \B{\mathcal{B}} %affine map
\def \C{\mathbb{C}} %complex plane
\def \Z{\mathbb{Z}} %integers
\def \SU{\mathbb{S}^1} %the unit circle
\def \re{\mathrm{Re}} %real part of complex number
\def \im{\mathrm{Im}} %real part of complex number
\newcommand {\closure}[1]{\textrm{Closure}(#1)} %closure of a set
\newcommand {\interior}[1]{\textrm{Interior}(#1)} %interior of a set

\def \vphi{\varphi} %the notation for the local charts

\def \O{\mathcal{O}} %big O notation
\def \S{\mathcal{S}} %surface (in R^3)
\def \K{\mathrm{\textbf{K}}} % vector of conformal distortions
\def \Tau{\mbox{\boldmath$\tau$}} % a vector of rotations tau_j
\def \F{\mathrm{\textbf{F}}} % face set
\def \E{\mathrm{\textbf{E}}} % edge set
\def \V{\mathrm{\textbf{V}}} % vertices set
\def \nei{\mathcal{N}} % \nei(j) the neighbors' indices of vertex i
\def \FF{\mathcal{F}} %the affine maps
\def \FK{\mathcal{F}_K} %the K-QC affine maps
\def \FKtau{\mathcal{F}_{K,\tau}} %maximal convex piece of \FK set around rotation theta using tau as parameter

\def \CD{\mathrm{\textbf{D}}}%conformal distortion
\def \DCD{\mathrm{\textbf{D}^d}} %discrete conformal distortion
\def \J{\mathrm{J}}%Jacobian

\def \-{\hspace{-2mm}}

\begin{document}

\title[]{Approximation of Polyhedral Surface Uniformization}%
\author{Yaron Lipman}%
\address{Weizmann Institute Of Science}%
\email{Yaron.Lipman@weizmann.ac.il}%

%\thanks{}%
%\subjclass{*}%

%\date{}%
%\dedicatory{}%
%\commby{}%
% ----------------------------------------------------------------
\begin{abstract}
We present a constructive approach for approximating the conformal map (uniformization) of a polyhedral surface to a canonical domain in the plane. The main tool is a characterization of convex spaces of quasiconformal simplicial maps and their approximation properties. As far as we are aware, this is the first algorithm proved to approximate the uniformization of general polyhedral surfaces.
\end{abstract}
\maketitle

% ----------------------------------------------------------------
\section{Introduction}

A polyhedral surface $\S$ is defined by stitching planar polygons along congruent edges. A polyhedral surface can be endowed with a conformal structure making it a Riemann surface \cite{stephenson2005introduction,Bobenko11}. The celebrated uniformization theory \cite{ahlfors2010conformal,farkas1992riemann} then implies the existence of a conformal map $\Psi:\S\too\T$ between $\S\subset \Real^d$  ($d$ is typically 3 in applications) and a topologically equivalent domain in the plane $\T\subset \C$. The main focus of this paper is building constructive approximations of this map for polyhedral surfaces with generally shaped polygonal faces. Topologically, we restrict our attention to disk-type surfaces. To date, we are not aware of any other existing algorithm that is proved to converge in the limit to the uniformization map for general polyhedral surfaces (i.e., with polygonal faces of arbitrary shape).

The problem of approximating \emph{planar} conformal mappings is considered well-understood and there is a wealth of methods that produce approximations to conformal mappings between planar domains \cite{Porter05historyand,papamichael2010numerical,driscoll2002schwarz}. Circle-packing \cite{stephenson2005introduction}, imitates conformal mappings by replacing infinitesimal circles with finite one, and was shown to converge to conformal mappings as the circles are refined \cite{Rodin_Sullivan_1987,He96onthe,He_Schramm98}. However, it seems no full generalization to polyhedral surfaces with generally shaped faces exists \cite{stephenson2005introduction}. Other constructions of ``discrete uniformization'' exist in the field of Discrete Differential Geometry (DDG) \cite{DDG_Oberwolfach08} where the focus is building a consistent and rich discrete theory. An example for such construction that can be used to compute discrete uniformization is by Springborn and coauthors \cite{Springborn:2008:CET:1360612.1360676,Bobenko10_CETM}. Discrete Ricci flow is another example \cite{jin2008discrete}. Other popular constructions can be found in \cite{Levy:2002:LSC:566654.566590,Gu:2003:GCS:882370.882388,Sheffer:2005:AFR:1061347.1061354,Kharevych:2006:DCM:1138450.1138461,ben2008conformal} however no proof of convergence to the uniformization map is provided for any of these methods so-far.

%of conformal mappings include application of the Schwarz-Christoffel formula %\cite{driscoll2002schwarz}. However, we are not aware of any uses or generalizations %of this formula to the polyhedral surface case.

In this paper we construct algorithms for approximating the uniformization map with guarantees (i.e. with convergence proof). Without loosing any generality we can subdivide each polygon in $\S$ into triangles and henceforth assume we have a triangulation $\S=(\V,\E,\F)$, where $\V=\set{v_i}\subset \Real^d$, is the set of vertices; $\E=\set{e_k}$, the set of edges; and $\F=\set{f_j}$, the set of oriented planar faces (triangles).

%\begin{floatingfigure}[r]{0.3\textwidth}
\begin{wrapfigure}{r}{0.2\textwidth}
  \begin{center}\vspace{-0.4cm}\hspace{-0.4cm}
    \includegraphics[width=0.2\textwidth]{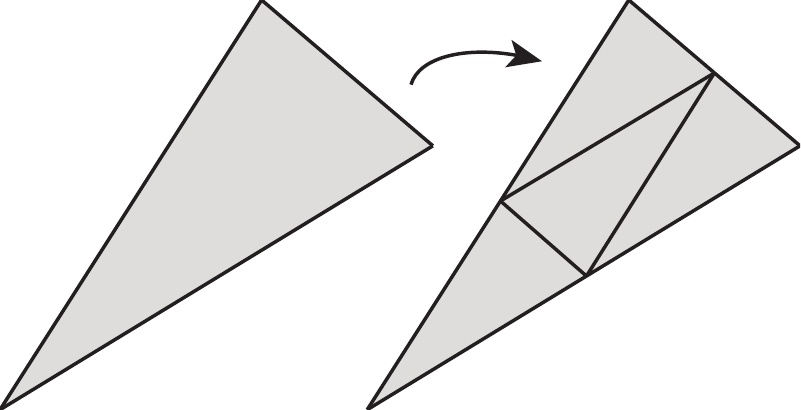}\vspace{-0.5cm}
  \end{center}
  %\caption{The Projection approach.}
  %\label{fig:multi_connected}
\end{wrapfigure}
We will construct successive approximations to $\Psi$ over a series of regular subdivisions of the surface $\S=\S^0\prec \S^1 \prec ...\prec \S^q$. By regular subdivision we mean that at each level we cut every face into four similar faces by connecting the mid-edges points, see inset figure, and Figure \ref{fig:cat_head} (top row). Our approximations will be simplicial mappings $\Phi^q\in\FF^{\S^q}$ of $\S^q$, namely piecewise-affine (over faces) and continuous mappings into the complex plane $\C$, the collection of such mappings over $\S^q$ will be denoted by $\FF^{\S^q}$. $\Phi^q$ will converge locally uniformly to $\Psi$. We will restrict our attention here to topological disks, and take as a canonical domain $\T$ the equilateral $\Delta(t_1,t_2,t_3)$ defined by its three corners $t_\ell=e^{\i (\ell-1)2\pi/3}, \ell=1,2,3$. To set a unique target uniformization map we will mark three distinct (positively oriented) boundary vertices $v_1,v_2,v_3 \in \V$ that will be mapped to the corners $t_1,t_2,t_3$ (respectively). This fixes all the degrees of freedoms of the map. Figure \ref{fig:cat_head} shows examples of simplicial approximations $\Phi^q$ to $\Psi:\S\too\T$ for a series of three refinements of a particular polyhedral surface: the middle row shows
planar checkerboard texture mapped by the inverse of the simplicial maps $\Phi^q$ to visualize the ``conformality'' of the approximations. The bottom row shows the homeomorphic image of $\S^q$ under $\Phi^q$ onto the equilateral domain. Bright-red color indicates high conformal distortion, while grey indicates low conformal distortion. Note that the approximations are improving as the mesh is refined. The main result on which we build upon when developing the algorithms in this paper  is:
\begin{figure}[t] %\vspace{-0.5cm}
\centering
\begin{tabular}{@{\hspace{0.0cm}}c@{\hspace{0.05cm}}c@{\hspace{0.0cm}}c@{\hspace{0.0cm}}c@{\hspace{0.0cm}}}
  \includegraphics[width=0.25\columnwidth]{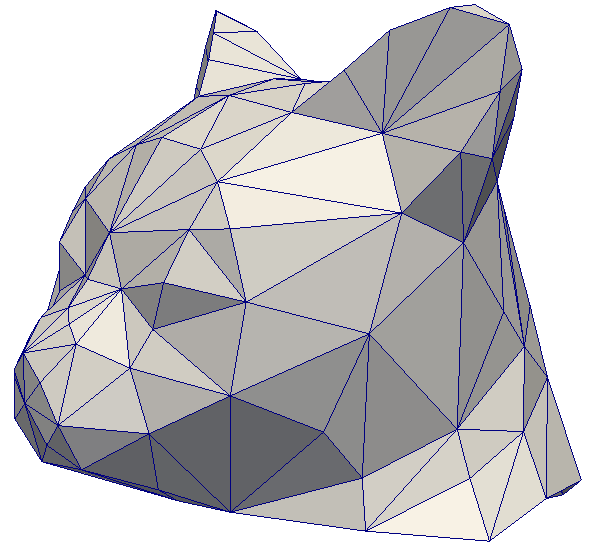}&
    \includegraphics[width=0.25\columnwidth]{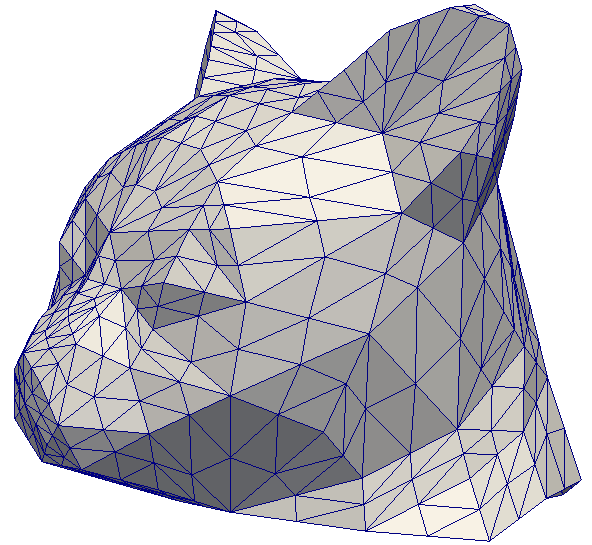}&
      \includegraphics[width=0.25\columnwidth]{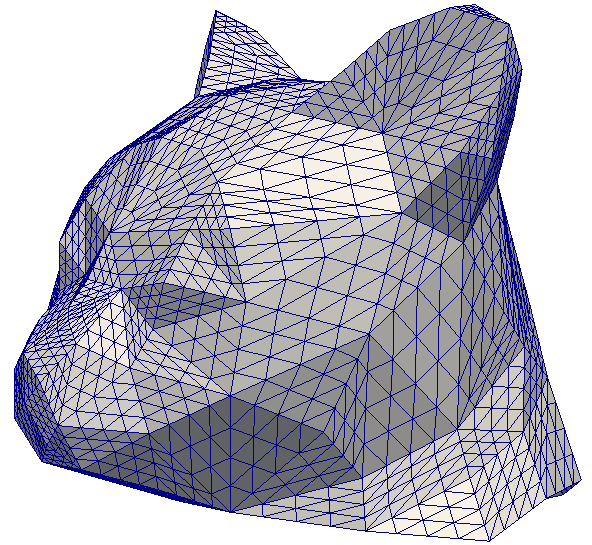}&
        \includegraphics[width=0.25\columnwidth]{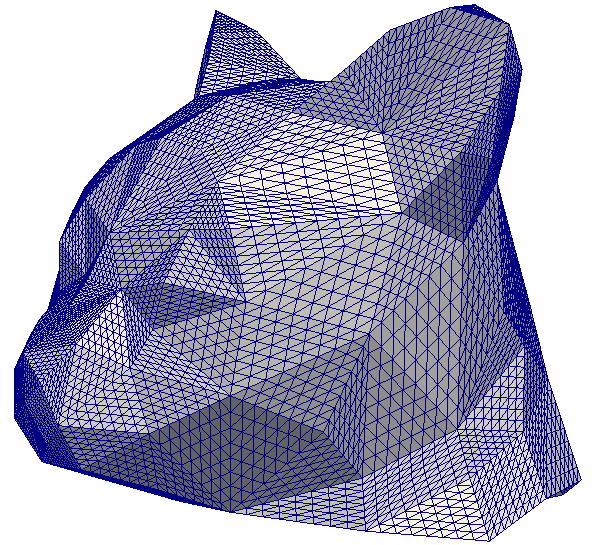}\\
&
    \includegraphics[width=0.25\columnwidth]{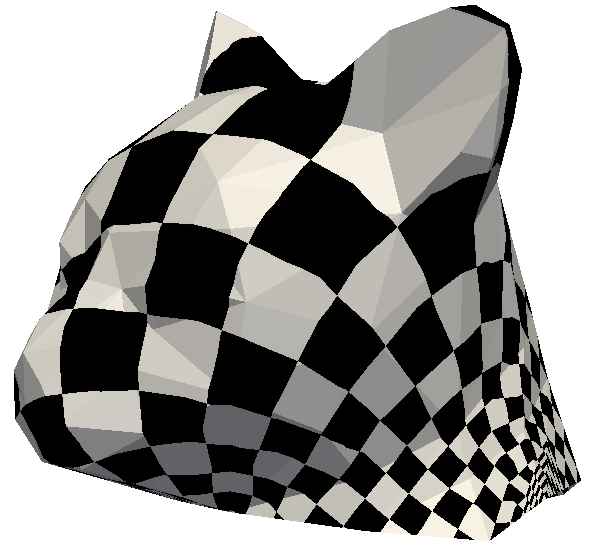}&
      \includegraphics[width=0.25\columnwidth]{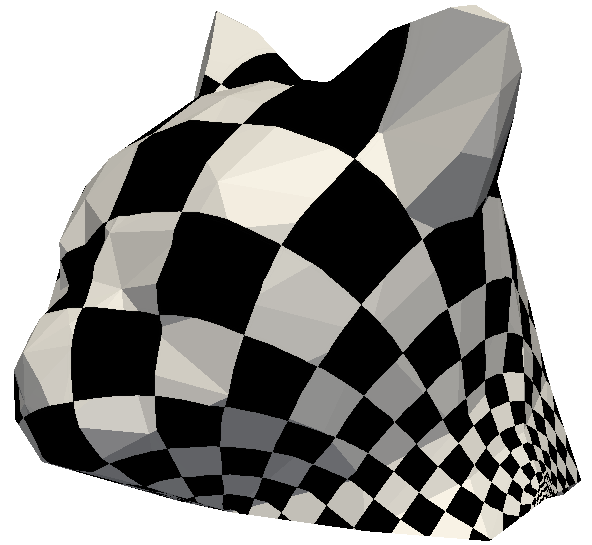}&
        \includegraphics[width=0.25\columnwidth]{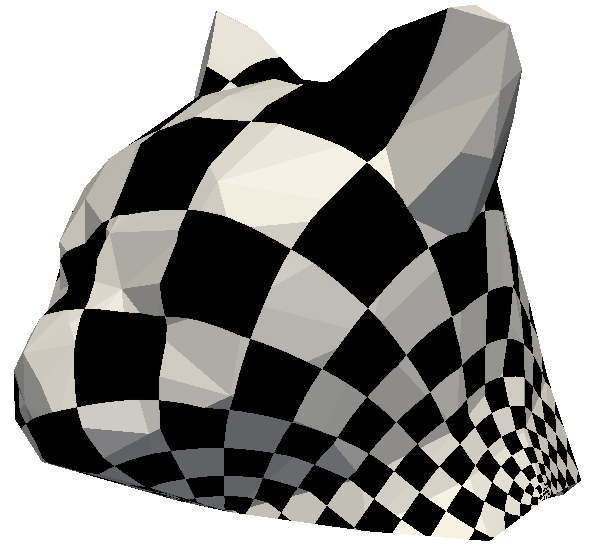}\\
&
    \includegraphics[width=0.25\columnwidth]{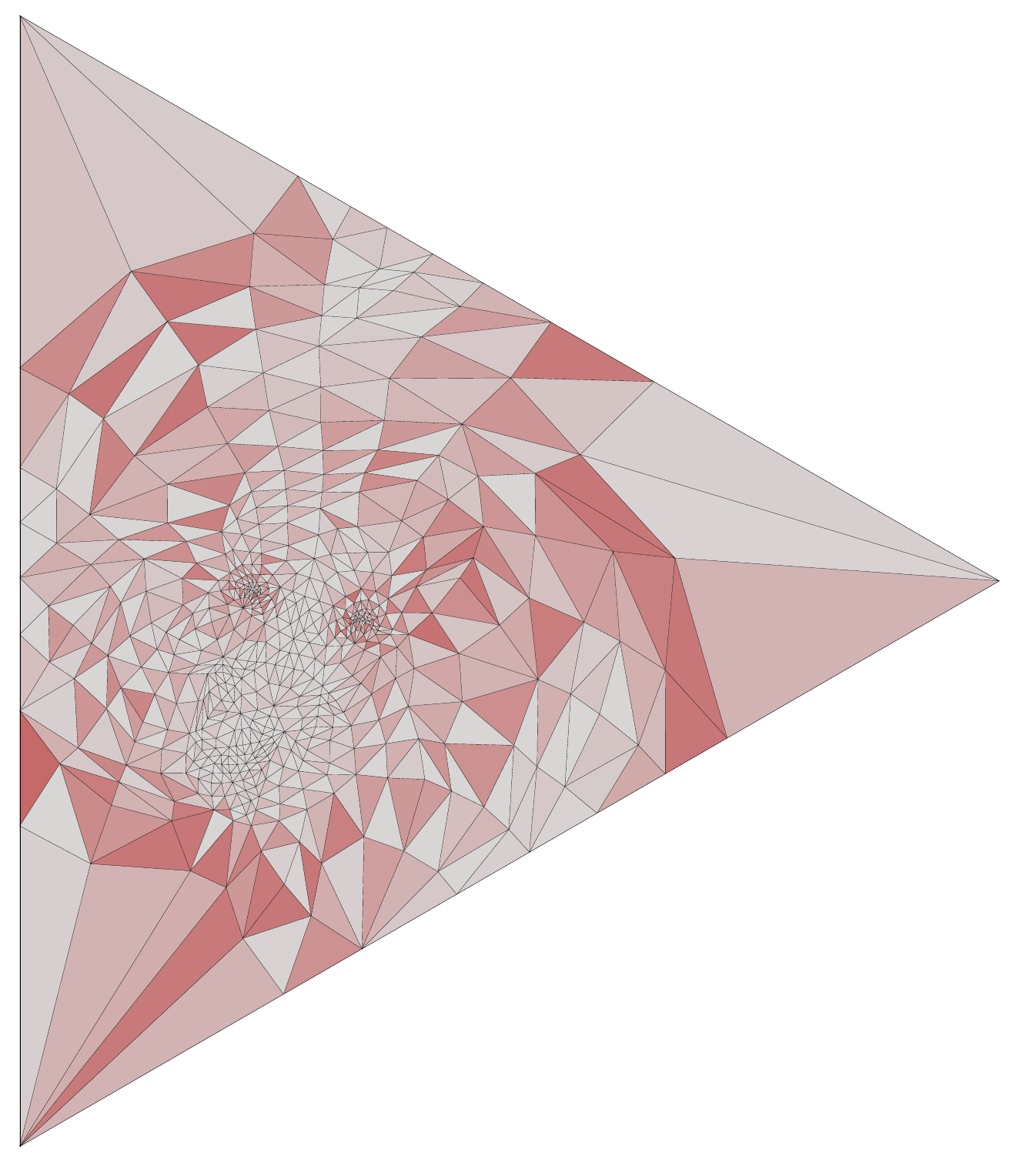}&
      \includegraphics[width=0.25\columnwidth]{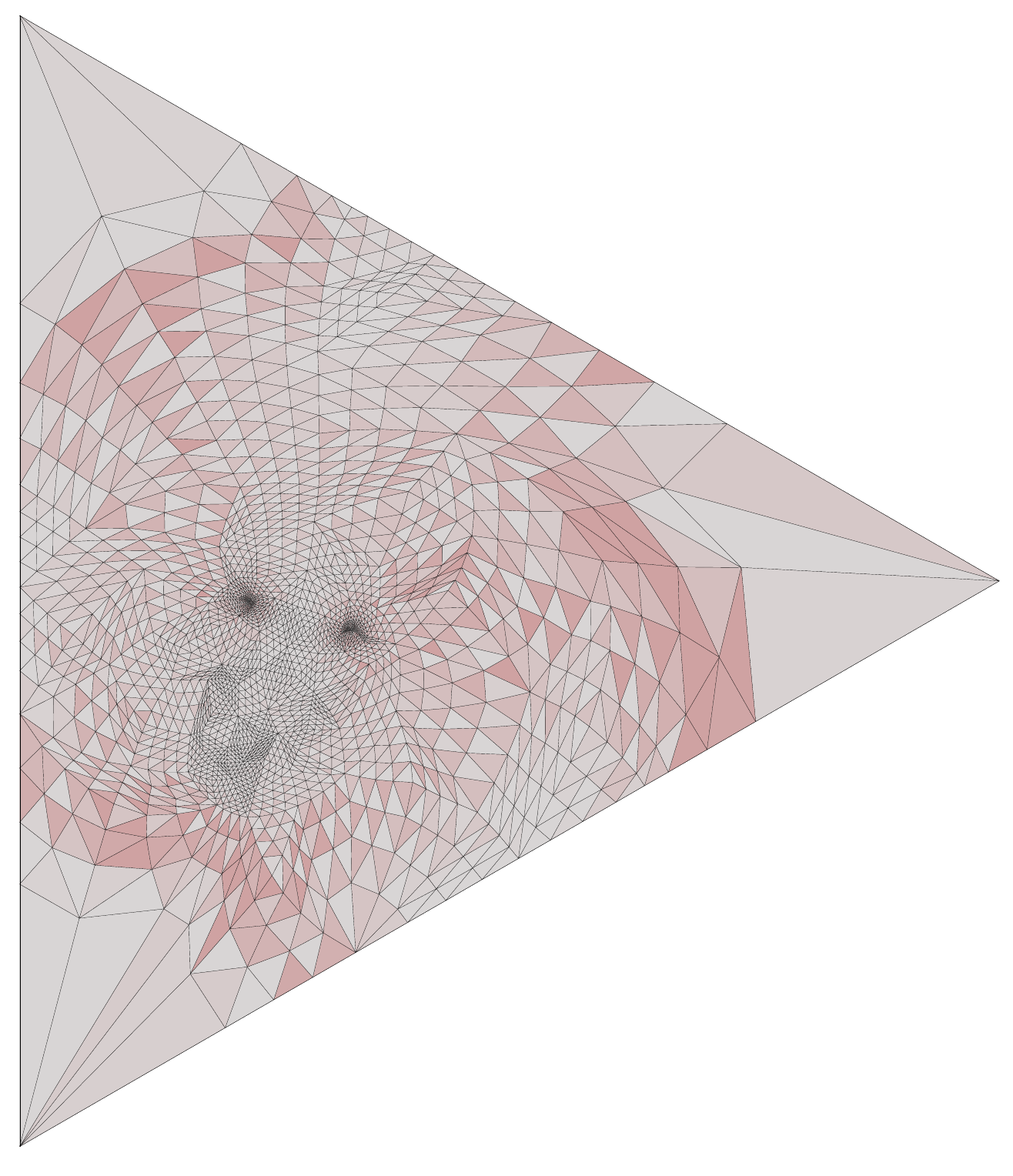}&
        \includegraphics[width=0.25\columnwidth]{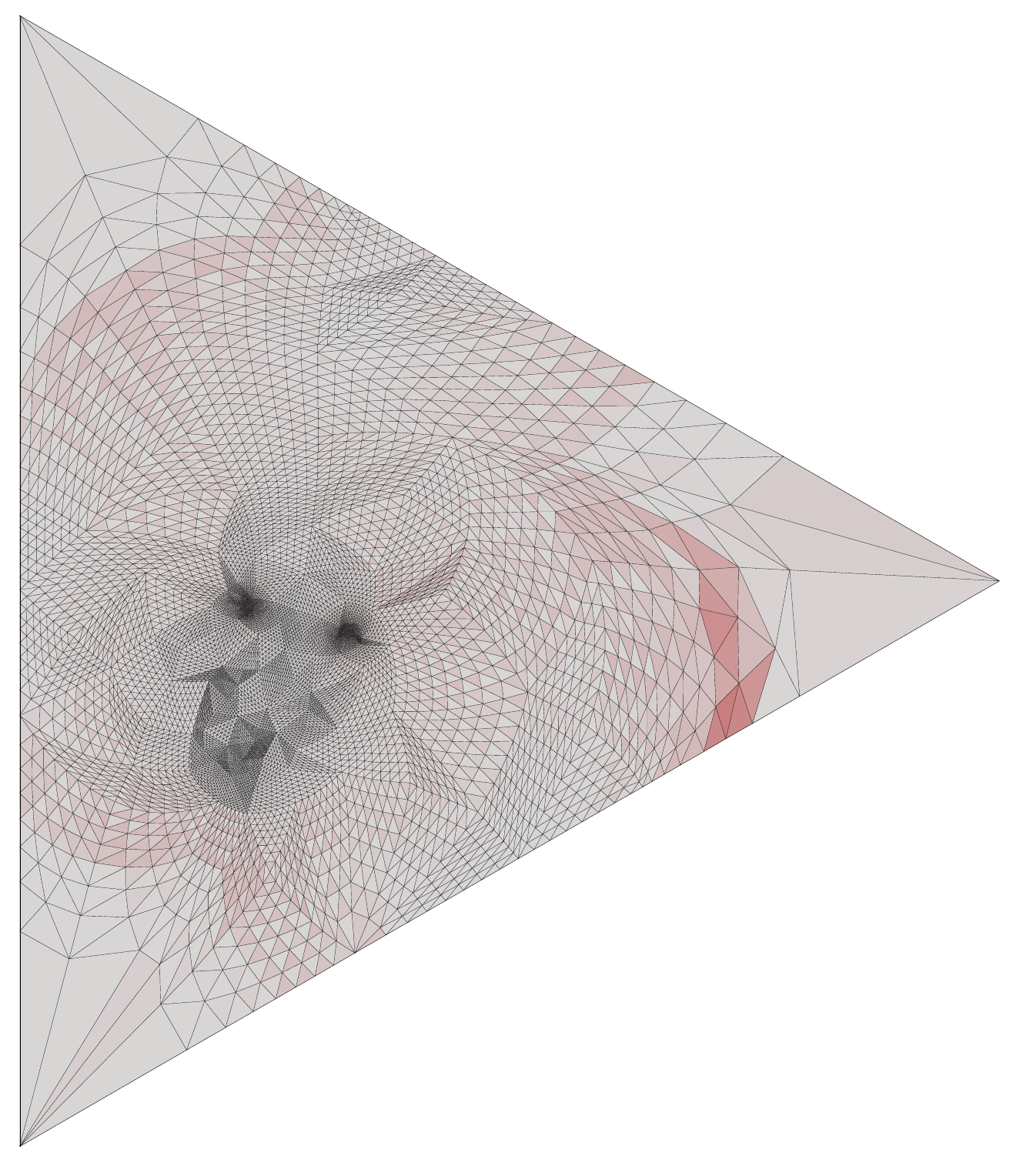}\\
$\S$ & $\S^1$ & $\S^2$ & $\S^3$ \\
\end{tabular}
  \caption{Approximation of the uniformization map with simplicial maps on subdivided versions of the origianl polyhedral surface $\S$.\vspace{-0.5cm}}\label{fig:cat_head}
\end{figure}

%\begin{figure}[h]
%  % Requires \usepackage{graphicx}
%  \includegraphics[width=0.6\columnwidth]{figures/cat_subdivision.png}\\
%  \caption{Approximation of the uniformization map with simplicial maps on subdivided versions of the origianl polyhedral surface $\S$.}\label{fig:cat_head}
%\end{figure}
\begin{thm}\label{thm:main}
Let $\Psi:\S\too\T$ be the uniformization map of a disk-type polyhedral surface to the equilateral $\T$ taking three prescribed boundary vertices of $\S$ to the corners of $\T$. Let $\S^q$ be the $q^{th}$-level subdivided version of $\S$. \\
If, for an arbitrary but fixed $\eps>0$, the argument of $\Psi'$ is known up to an error of $\pm(\frac{\pi}{2}-\eps)$, then one can construct a series of non-empty \emph{convex} spaces $U^q=\set{\Phi^q}\subset \FF^{\S^q}$ of simplicial maps of $\S^q$ such that:

\begin{enumerate}
\item
Every map $\Phi^q \in U^q$ is $K$-quasiconformal (QC) homeomorphism that maps $\S$ onto $\T$, with some constant $K$ independent of $q$.

\item
Every series $\set{\Phi^q}_{q\geq Q}$, where $\Phi^q\in U^q$, converges locally uniformly to the uniformization map $\Psi$. That is, $\Phi^q\circ\Psi^{-1}:\T\too\T$ converges uniformly in any compact subset of $\interior{\T}$ to the identity map $I_d:\T\too\T$.
\end{enumerate}
\end{thm}

Let us clarify the assumption ``the argument of $\Psi'$ is known up-to an error of $\pm(\frac{\pi}{2}-\eps)$''. What we mean by that is that at every point $p\in \S$, we can choose (arbitrary) chart $(\vphi_i,\Omega_i)$, $\vphi_i:\Omega_i\subset\S\too\C$, where $p\in\Omega_i$, $z=\vphi_i(p)$, and that we can point an angle $\tau$ in the range $$\parr{\arg \brac{ (\Psi\circ\vphi_i^{-1})'(z) }-\parr{ \frac{\pi}{2}-\eps},\arg \brac{(\Psi\circ\vphi_i^{-1})'(z) }+\parr{\frac{\pi}{2}-\eps}},$$
for $\eps>0$ arbitrary small but fixed. Intuitively, if we know in which "half" of $\SU$ ($\SU$ denotes the unit circle) the argument of the derivative of the map we are looking for resides in, then we can approximate $\Psi$ via a convex program, namely looking for an element in a \emph{known} convex subset of the simplicial maps $U^q \subset \FF^{\S^q}$.

This theorem will be proved in several parts: first, in Section \ref{s:feasibility} we will show that $\FF^{\S^q}$ contains at-least one simplicial map that is: 1) quasiconformal, namely a homeomorphism with bounded conformal distortion, 2) its conformal distortion is converging to 1 with almost linear rate. Second, in Section \ref{s:approximation}, we will show that \emph{any} series of quasiconformal simplcial maps $\Phi^q$ that their conformal distortion converge to 1 converges to the uniformization map $\Psi$ as described in Theorem \ref{thm:main}. In Section \ref{s:algorithm} we will study the space of $K$-quasiconformal simplicial maps $\FF^{\S^q}_K\subset \FF^{\S^q}$ and characterize a collection of convex subsets $\set{\FF^{\S^q}_{K,\Tau}}$, $\Tau=(\tau_1,..,\tau_{|\F^q|}), \tau_j\in \SU$. We will further show that given an $\frac{\pi}{2}-\eps$ approximation of the argument of $\Psi'$, one can single out one of these convex subsets $\FF^{\S^q}_{K,\Tau^*}$ that will be non-empty. This will finish the proof of Theorem \ref{thm:main}.

Building upon Theorem \ref{thm:main} we will suggest two algorithms. The first algorithm is exhaustive but theoretically fully justified: it will test many candidates $\Tau$, one of which is guaranteed to lead to a non-empty space $\FF^{\S^q}_{K,\Tau}$. The number of candidates will be rather large but shown to be independent of $q$ (i.e., exponential in the number of original faces $|\F|=|\F^0|$). Due to the large number of candidates $\tau$ this algorithm will mainly have theoretical importance but only limited practical applicability. Nevertheless, as far as we are aware this is the first algorithm to approximate the uniformization map for general faced polyhedral surfaces.

The second algorithm will be greedy in nature: it will start with some arbitrary convex space $\FF^{\S^q}_{K,\Tau^0}$, and will iteratively move to other convex spaces $\FF^{\S^q}_{K,\Tau^1},\FF^{\S^q}_{K,\Tau^2},...$ When terminating successfully it guarantees an approximation to the uniformization map, enjoying all the properties of the first algorithm (and Theorem \ref{thm:main}). However, the drawback here is that we do not have a proof that it will always ends up with a non-empty space (i.e., terminate successfully). Nevertheless, it works well in practice and is much more computationally efficient than the (first) exhaustive algorithm.

We start by defining the smooth and discrete conformal structures of $\S$, set notations and a few preliminary lemmas.

\section{Smooth and discrete conformal structures}\label{s:smooth_and_discrete_conformal_structures}
As noted above, a polyhedral surface $\S$ admits a smooth (classical) conformal structure. We start by defining it.

%\begin{floatingfigure}[r]{0.3\textwidth}
\begin{wrapfigure}{r}{0.45\textwidth}
  \begin{center}
    \includegraphics[width=0.45\textwidth]{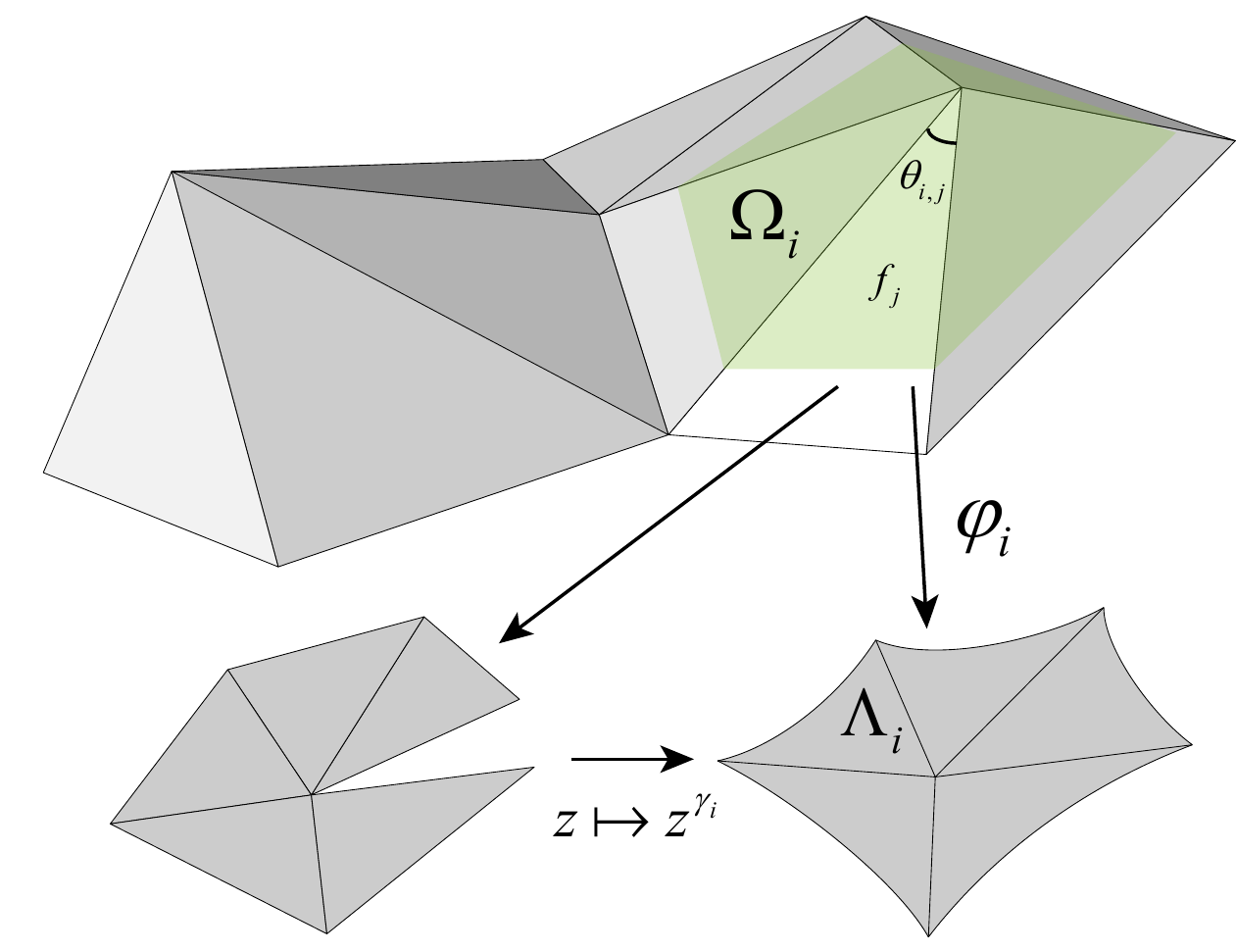}%\vspace{-0.2cm}
  \end{center}
  %\caption{The Projection approach.}
  %\label{fig:multi_connected}
\end{wrapfigure}
%\end{floatingfigure}
A conformal structure is defined by providing a conformal atlas, that is, a set of analytic coordinate charts \cite{farkas1992riemann}. We will define such an atlas already customized to our later constructions. We fix a constant $0 < \zeta < 1/4$, and for each vertex $v_i\in\V$, denote by $\nei_i$ the set of indices of 1-ring neighbors to vertex $i$.

We distinguish three types of vertices in our surface $\S$: interior, boundary, and corner. Interior vertices are vertices in the interior of $\S$, boundary are on the boundary but not one of the three corners $v_1,v_2,v_3$.

\paragraph{\textbf{Interior charts}} Let $v_i\in\V$ be an interior vertex. Set $\Theta_i=2\pi$. Set the domain $\Omega_i\subset \S$ to be the interior of the convex hull (over the surface) of the points $\zeta v_i + (1-\zeta)v_{i'}$, where $i'\in\nei_i$ (see inset, where $\Omega_i$ for one particular vertex is colored in green). We define the chart $\vphi_i:\Omega_i\too \C$ by first rigidly unfolding each triangle in $\Omega_i$ onto the plane, taking $v_i$ to the origin, and second, composing each (now planar) triangle with the map $z\mapsto z^{ \gamma_i }$, $\gamma_i=\frac{\Theta_i}{\theta_i}$ where $\theta_i=\sum_j \theta_{i,j}$ is the sum of angles at vertex $v_i$, where $\theta_{i,j}$ is the angle of face $j$ (adjacent to vertex $i$) at vertex $i$. This composition is made (possibly by incorporating rigid transformations in the plane) such that $\vphi_i$ is globally continuous, see \cite{stephenson2005introduction,Bobenko11} for more details and the inset for an illustration. Let us denote $\Lambda_i=\vphi_i(\Omega_i)\subset\C$.

%We will map the surface $\S$ onto the unit equilateral $\T=\Delta(t_1,t_2,t_3)$ defined as the convex hull of the triplet $t_\ell=e^{\i(\ell-1)2\pi/3}\in \C, \ell=1,2,3$, that is $\set{\sum_\ell \lambda_\ell t_\ell \mid \lambda_\ell\geq0, \sum_\ell \lambda_\ell=1}$.

\paragraph{\textbf{Boundary charts}} Let $v_i\in\V$ be a boundary vertex ($i\ne 1,2,3$). Set $\Theta_i=\pi$. The domain $\Omega_i$ is defined as the interior of the convex hull of $v_i$ and $\zeta v_i + (1-\zeta)v_{i'}$, $i'\in\nei_i$, and  $\vphi_i:\Omega_i\too\C$ is defined via mapping the neighborhood congruently to the plane as before, composing it with $z\mapsto z^{\gamma_i}$, where $\gamma_i=\frac{\Theta_i}{\theta_i}$.

\paragraph{\textbf{Corner charts}} For the three corner vertices $v_1,v_2,v_3$ we set $\Theta_i=\frac{\pi}{3}$ and define $\Omega_i$, $i=1,2,3$, similar to the boundary vertices' charts with the exception of using the mapping  $z\mapsto z^{\gamma_i}$, where $\gamma_i=\frac{\Theta_i}{\theta_i}$.

All the transition maps $\vphi_{i'}\circ\vphi_i^{-1}$ are conformal and therefore define a conformal structure over $\S$ :

% Let us prove a simple but useful lemma regarding the charts:
% \begin{lem}\label{lem:extend_charts}
% There exists a constant $\kappa>0$, such that for all charts $(\Omega_i,\vphi_i)$
% \end{lem}

\begin{lem}\label{lem:transition_maps_holo_and_bounded_derivative}
The transition maps $\vphi_{i,i'}=\vphi_i\circ \vphi_{i'}^{-1}$ are holomorphic.
\end{lem}
\begin{proof}
The holomorphy of the transition maps $\vphi_{i,i'}=\vphi_i\circ \vphi_{i'}^{-1}$ can be understood (see e.g.,\cite{Bobenko11}) from the fact that these maps are compositions of similarities $z\mapsto \alpha z+ \delta$ and the analytic maps (note that we avoid the origin) $z\mapsto z^{\vartheta\pi/\theta_i}$, $\vartheta=1/3,1,2$. The holomorphy across edges can be verified using standard extension theorems of conformal maps.
%The fact that the derivatives can be bounded by a universal constant can be explained as follows. There exists some positive $\eps>0$ such that each $\vphi_{i,i'}$ can be extended to $\D(z',\eps)=\set{z \mid |z-z'|<\eps }$ (we denote $\D(\xi,r)=\set{z\mid \abs{z-\xi}<r}$, $\D=\D(0,1)$) around every point $z'\in\vphi_{i'}(\Omega_{i'}\cap\Omega_{i})$. We can also find a constant $M>0$ sufficiently large such that $\abs{\vphi_{i,i'}(z')}\leq M$ for all $z'\in \cup_{z\in\vphi_{i'}(\Omega_{i'}\cap\Omega_{i})}\D(z,\eps)$. From Cauchy theorem
%\begin{eqnarray*}
%\abs{\vphi'_{i,i'}(z')}&=& \frac{1}{2\pi}\int_{\partial\D(z',\eps)}\abs{\frac{\vphi_{i,i'}(\xi)}{(\xi-z')^2}}\abs{d\xi}\\
%&\leq&\frac{1}{2\pi}M \frac{2\pi \eps}{\eps^2}\\
%&=&\frac{M}{\eps}.
% \end{eqnarray*}
\end{proof}

Now that we have a (smooth,classical) conformal structure, the notion of conformal mappings from $\S$ to the plane $\C$ is well-defined; a map $\Psi:\S\too \C$ is conformal if for every chart $(\Omega_i,\vphi_i)$, the map $\psi_i=\Psi \circ \vphi_i^{-1}$ defined over $\Lambda_i=\vphi_i(\Omega_i)$ is conformal in the classical sense.

The main object of this paper is to approximate the conformal mapping $\Psi$ that maps the polyhedral surface bijectively to the equilateral triangle $\T$. The existence and uniqueness of such a mapping is set by the uniformization theorem \cite{ahlfors2010conformal,farkas1992riemann}:
\begin{thm}
There exists a homeomorphism $\Psi:\S\too \T$ that is conformal in the interior of $\S$, $\Psi:\interior{\S}\too \interior{\T}$\footnote{$\interior{Q}$ denotes the interior of the set $Q$.}.  The map $\Psi$ is uniquely set once required to take $v_1,v_2,v_3$ to the corners $t_1,t_2,t_3$ of $\T$.
\end{thm}

We will approximate $\Psi$ by constructing quasiconformal simplicial mappings $\Phi^q:\S^q\too \T$ from subdivided versions of $\S$ to $\T$.

%\begin{floatingfigure}[l]{0.3\textwidth}
\begin{wrapfigure}{l}{0.3\textwidth}
  \begin{center}
    \includegraphics[width=0.3\textwidth]{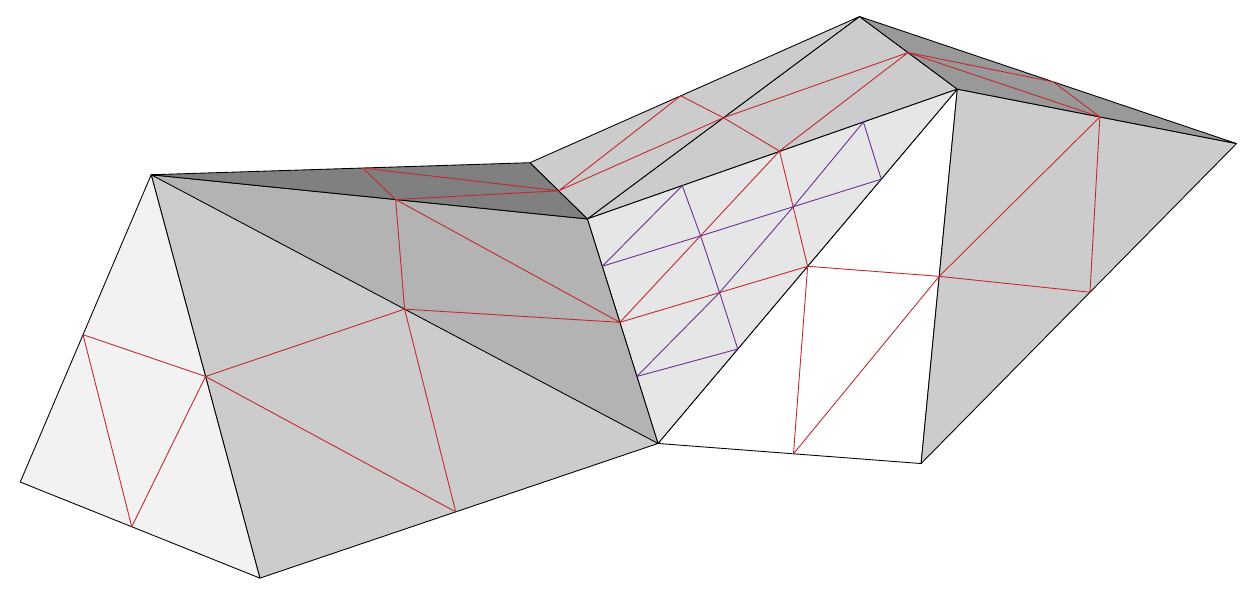}%\vspace{-0.2cm}
  \end{center}
  %\caption{The Projection approach.}
  %\label{fig:multi_connected}
\end{wrapfigure}
%\end{floatingfigure}
The subdivided triangulations $\S^q=\parr{\V^q,\E^q,\F^q}, q=0,1,2,...$ of $\S$ are constructed by the standard $1-4$ subdivision rule in each triangle, where $\S^0=\S$. For example, in the inset the red mesh shows $\S^1$, and we show in purple all the faces $f_{j'}\in \F^2 \cap f_j $ inside one face $f_j\in\F$ of $\S$.

We now turn to define the \emph{discrete} conformal structure over $\S^q$. A discrete conformal structure is basically assigning angles to the corners of each triangle such that the angle sum of each triangle is $\pi$. Equivalently, we can embed each triangle in the Euclidean plane and think of it up-to a similarity transformation. A \emph{consistent} discrete conformal structure will approximate the smooth one as the surface is refined. A simple way to do it is by mapping each triangle to the plane with the charts $\vphi_i$ and taking its image Euclidean triangle to define its conformal structure. Obviously, using different atlas, or assignments of triangles to charts will lead to different discrete conformal structure, but as we will show, at the limit it won't matter. So we will use the Atlas defined above, and describe an arbitrary assignment of faces to charts in $\F^q$.

We start by associating, for $q = 2$, each triangle $f_j\in\F^2$ to some chart $\set{\Omega_i,\vphi_i}$, $i=i_j$, by making sure that $\closure{f_j}\subset \Omega_i$.

%\begin{floatingfigure}[r]{0.2\textwidth}
\begin{wrapfigure}{r}{0.2\textwidth} %\vspace{-0.5cm}
  \begin{center}
    \includegraphics[width=0.2\textwidth]{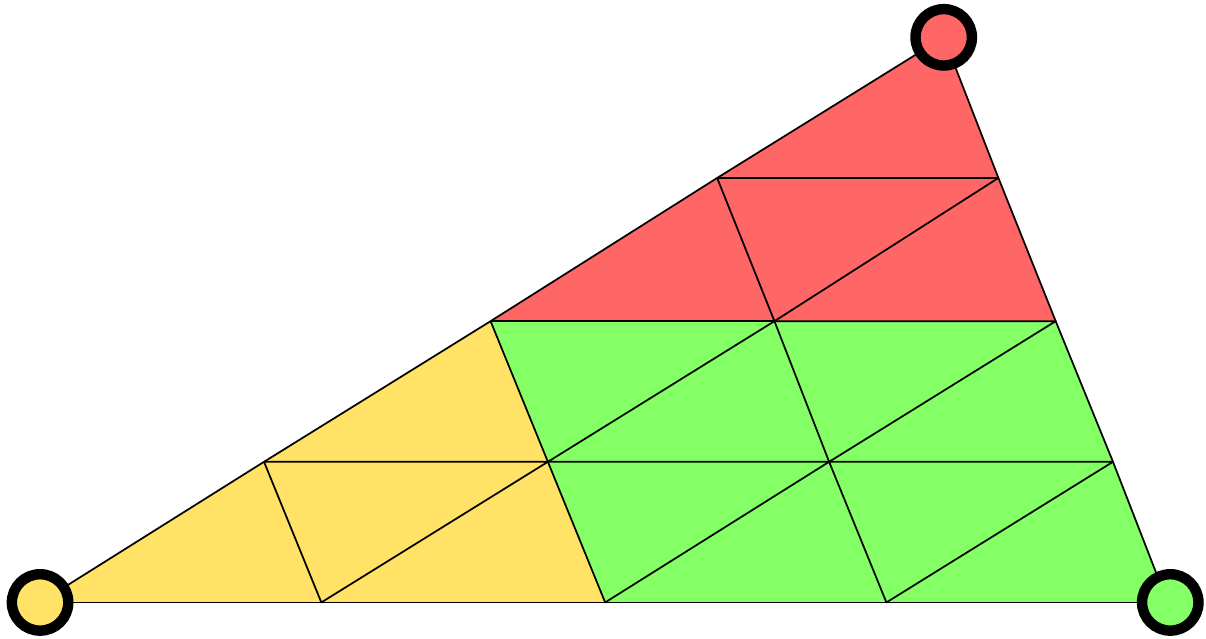}%\vspace{-0.2cm}
  \end{center}
  %\caption{The Projection approach.}
  %\label{fig:multi_connected}
\end{wrapfigure}
%\end{floatingfigure}
For example, we can use the assignment rule shown in the inset figure (each color indicates association of that colored face to a different vertex of the triangle and hence to a difference chart).

%\begin{floatingfigure}[l]{0.3\textwidth}
\begin{wrapfigure}{r}{0.5\textwidth}
  \begin{center}
    \includegraphics[width=0.5\textwidth]{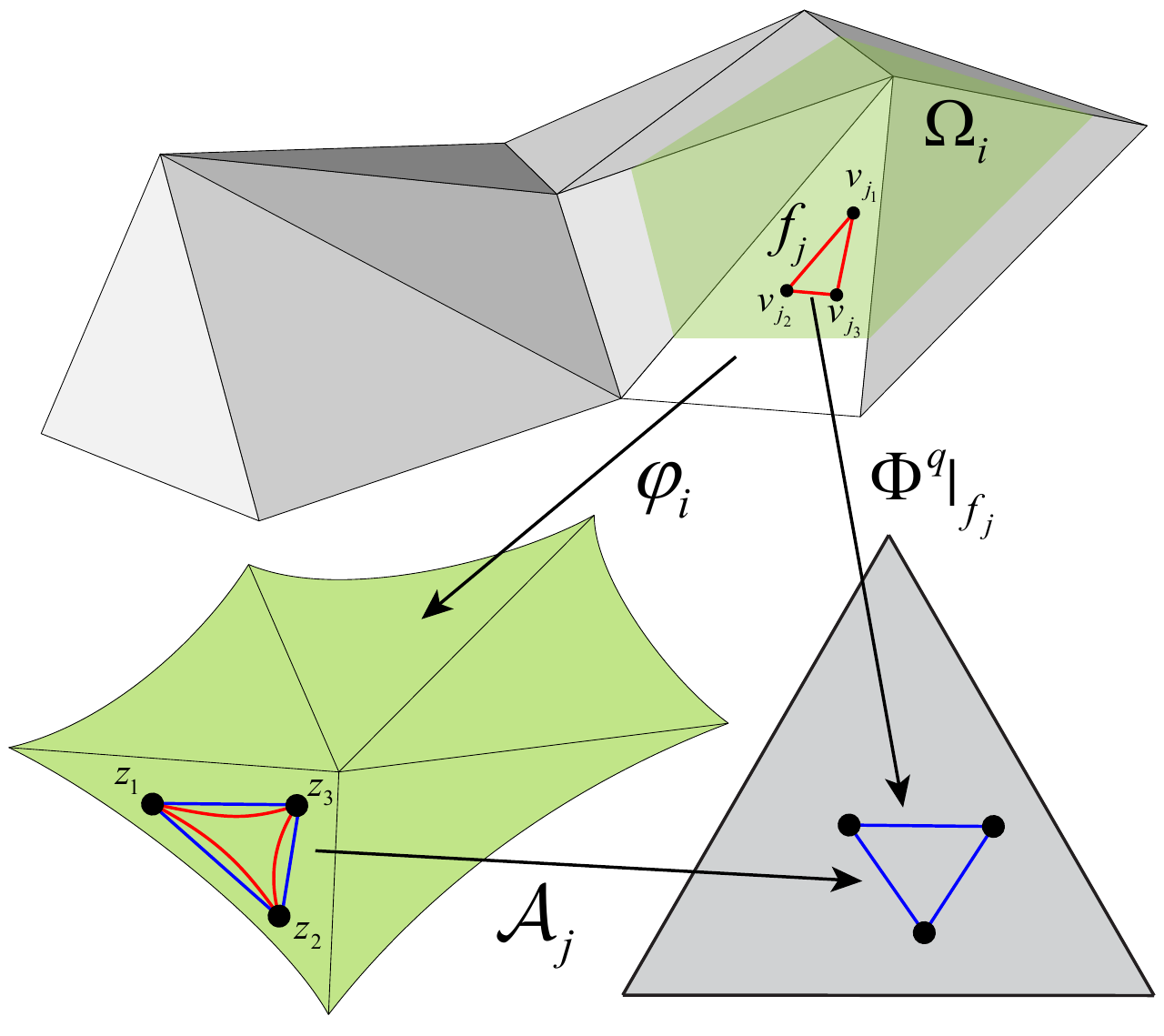}%\vspace{-0.2cm}
  \end{center}
  %\caption{The Projection approach.}
  %\label{fig:multi_connected}
\end{wrapfigure}
%\end{floatingfigure}
For any $f_{j}\in \F^{q}, q> 2$ we associate the chart based on the triangle's (unique) ancestor in level $q=2$. We define the discrete conformal structure for each face $f_j\in \F^q, q\geq 2$ by mapping the face's vertices $v_{j_1},v_{j_2},v_{j_3}$ to the plane, that is $z_\ell = \vphi_i(v_{j_\ell}), \ell=1,2,3$, $\set{\Omega_i,\vphi_i}$ the associated chart $i=i_j$, and using the Euclidean triangle $\Delta(z_1,z_2,z_3)$ to define the discrete conformal structure of $f_j$. The inset below shows in blue (bottom-left) the triangle $\Delta(z_1,z_2,z_3)$.

Once we defined the discrete conformal structure we can define the \emph{discrete conformal distortion} of a simplicial map $\Phi^q:\S^q\too\C$. We look per face $f_j\in\F^q$, $q\geq 2$, and denote by $\A_j$ the affine map mapping the triangle $\Delta(z_1,z_2,z_3)$ (remember that $z_\ell=\vphi_i(v_{j_\ell})$) to the triangle $\Delta(\Phi^q(v_{j_1}),\Phi^q(v_{j_2}),\Phi^q(v_{j_3}))$. We will also refer to $\A_j$ as $\Phi^q\mid_{f_j}$ expressed ``in the local coordinate chart''. Then, the discrete conformal distortion $\DCD(\Phi^q\mid_{f_j})$ of $\Phi^q\mid_{f_j}$ is defined by $$\hspace{-9cm}\DCD(\Phi^q\mid_{f_j})=\CD(\A_j),$$
where $\CD(\A_j)$ denotes the standard conformal distortion of the planar affine map $\A_j$, namely, $\CD(\A_j)=\frac{\Sigma_j}{\sigma_j}$, the ratio of the larger to smaller singular values $\Sigma_j\geq\sigma_j\geq 1$ of the linear part of the planar affine map $\A_j$. See the inset figure for an illustration.

The discrete conformal distortion of the full simplicial map $\Phi^q$ is accordingly defined by
\begin{equation}\label{e:discrete_conformal_distortion}\hspace{-9cm}
  \DCD(\Phi^q)=\max_{f_j \in \F^q}\DCD(\Phi^q\mid_{f_j}).
\end{equation}
For later use, we denote by $\DCD_{\Omega}(\Phi^q)$, $\Omega\subset \S$, the discrete conformal distortion of the map $\Phi^q$ restricted only to triangles contained in the set  $\Omega$, that is
\begin{equation}\label{e:discrete_conformal_distortion_Omega}%\hspace{-5cm}
    \DCD_\Omega(\Phi^q)=\max_{f_j \in \F^q , f_j\subset \Omega}\DCD(\A_j).
\end{equation}
Let us denote by $\FF^{\S^q}$ the space of simplicial maps (i.e., continuous  piecewise affine) mapping the triangulation $\S^q$ to the plane, satisfying the boundary constraints of mapping $\partial \S^q$ ($\partial\S^q$ denotes the boundary of the polyhedral surface $\S^q$) bijectively to $\partial \T$ and taking the corner vertices $v_1,v_2,v_3$ to the triangle's corners $t_1,t_2,t_3$. Further define the subset $\FF^{\S^q}_\K \subset \FF^{\S^q}$, $\K=(K_1,K_2,...,K_{|\F^q|})$ to include only orientation-preserving homeomorphisms $\Phi^q\in \FF^{\S^q}_\K$ such that $\DCD(\Phi^q\mid_{f_j})\leq K_j$. Denote $|\K|=\max{K_j}$, and by $\FF^{\S^q}_K$ we will mean $\cup_{\abs{\K}\leq K}\FF^{\S^q}_\K$, that is, simplicial maps with maximum bound $K$ on its conformal distortion per face.

There are two technical issues to be taken care-of for later constructions.
First, let us highlight a small technical property of the charts and the association rule of faces we have defined that will be used in Section \ref{s:feasibility} when we approximate the uniformization map with simplicial map: we show that there exists some positive gap between the triangles and the boundary of the charts' they are associated with. Denote the disk $\D(z,r)=\set{w\mid \abs{w-z}<r}$.
\begin{lem}\label{lem:epsilon_distance_to_charts_boundary}
There exists some constant $\eps_0>0$, such that $$\cup_{p\in f_j}\D(\vphi_i(p),\eps_0)\subset \Lambda_i=\vphi(\Omega_i),$$ for all $f_j\in\F^q, q\geq2$, associated with chart $\set{\Omega_i,\vphi_i}$, $i=i_j$
\end{lem}
\begin{proof}
It is enough to prove the inclusion for $q=2$. Take arbitrary face $f_j\in\F=\F^0$, and one of its vertices $v_i\in\V=\V^0$. Since when defining $\Omega_i$ we took $0 < \zeta < 1/4$ we see by the association procedure that every triangle $f_{j'}\in\F^q$ is far from $\partial\Omega_i$ at least by some constant $\eps_j>0$. When applying $\vphi_i$, since it has bounded derivatives (when considering it as map from the Euclidean faces to the complex plane) away from the vertex $v_i$, we have some other $\eps'_j>0$ bounding this distance from below. Taking the minimum $\eps'_j$ over the (finite) set of faces in $\F$ and their vertices we finish the argument.
\end{proof}

A second issue is that we need to make sure that approximating the power maps $h:z\mapsto z^\gamma$, $\gamma>0$ with simplicial maps (via sampling at the vertices and extending linearly) are $K$-quasiconformal with some universal bound on their distortion $K\geq 1$, independent of subdivision level $q$. For an example of a simplicial power map see Figure \ref{fig:power_map}.  \begin{figure}[h]
    \includegraphics[width=0.7\columnwidth]{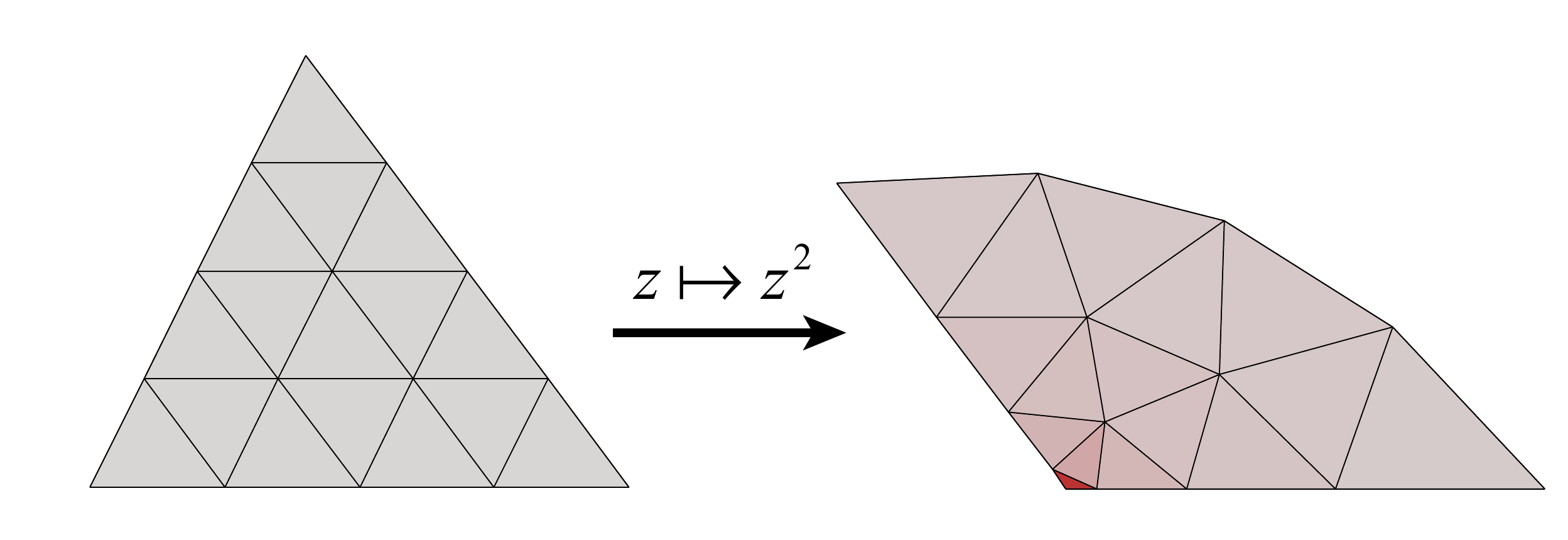}%\vspace{-0.2cm}\\
  \caption{The simplicial map $h^q,q=2$ sampled from $h(z)=z^2$.}\label{fig:power_map}
\end{figure}

We conjecture the following:
\begin{conj}\label{con:discretization_of_power_map}
Let $T=\Delta(\xi,0,\eta)$ be a triangle and denote the angle $\theta=\measuredangle(\xi,0,\eta)$. Further let $T^q$ be the $q^{th}$ level of regular 1-4 subdivision of $T$. Denote by $h(z)=z^\gamma$ the power map, and assume that $\gamma \theta < \pi$. Then, the simplicial maps $h^q$ defined by sampling $h(z)$ over the vertices of $T^q$ and extending by linearity are homeomorphisms that satisfy $\CD(h^q)\leq K$ for some $K\geq 1$ independent of $q$.
\end{conj}

For our needs it is enough to prove the following, slightly weaker, result:

\begin{lem}\label{lem:discretization_of_power_map}
Let $T=\Delta(\xi,0,\eta)$ be an isosceles triangle ($|\xi|=|\eta|$) and denote the angle $\theta=\measuredangle(\xi,0,\eta)$. Further let $T^q$ be the $q^{th}$ level of regular 1-4 subdivision of $T$. Denote by $h(z)=z^\gamma$ the power map, and assume that $\lceil \gamma \rceil \theta < \frac{\pi}{2}$, and that $\theta < 60.4^\circ$. Then, the simplicial maps $h^q$ defined by sampling $h(z)$ over the vertices of $T^q$ and extending by linearity are homeomorphisms that satisfy $\CD(h^q)\leq K$ for some $K\geq 1$ independent of $q$.
\end{lem}
The proof for this lemma is rather technical and therefore deferred to Appendix \ref{a:auxilary_lemmas}.
\begin{floatingfigure}[r]{0.4\textwidth}
%\begin{wrapfigure}{r}{0.2\textwidth}%\vspace{-0.5cm}
  \begin{center}
    \includegraphics[width=0.4\textwidth]{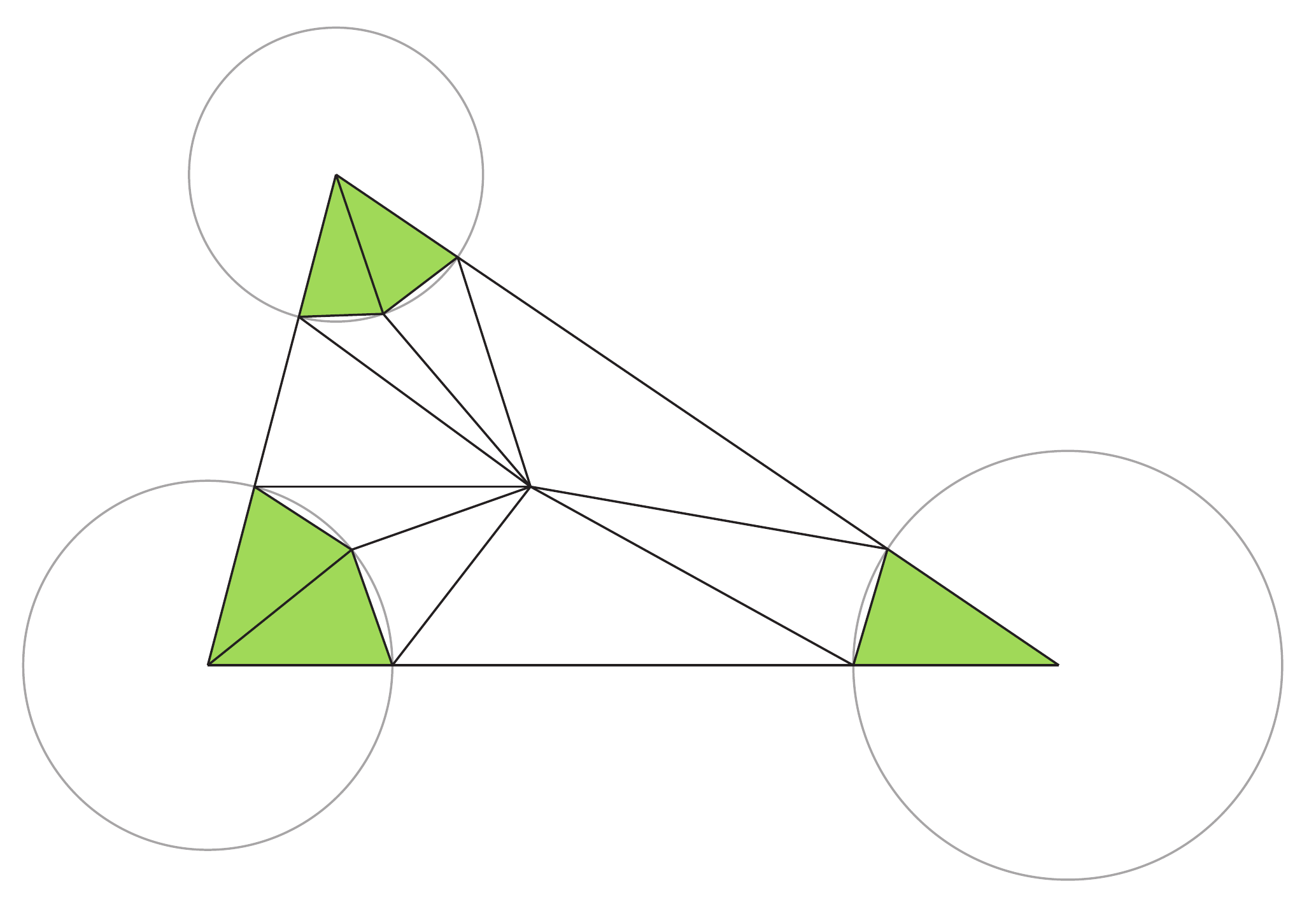}%\vspace{-0.2cm}
  \end{center}
  %\caption{The Projection approach.}
  %\label{fig:multi_connected}
%\end{wrapfigure}
\end{floatingfigure}

Building upon this lemma, we can (without loosing generality) subdivide each of the original triangles $f_j\in\F$ of the polyhedral surface $\S$, so that the two conditions $\lceil\gamma\rceil\theta_{i,j}<\frac{\pi}{2}$, and $\theta_{i,j}<60.4^\circ$, where $\gamma=\frac{\Theta_i}{\theta_i}$ are satisfied. We can also guarantee that every triangle that touches one of the original vertices $\V$ is isosceles (all the other triangles have all flat vertex angles\footnote{By flat vertex angles we mean the angle sum around a vertex is $2\pi$.}  and so their charts are rigid congruencies and hence quasiconformal). Such subdivision is shown in the inset. Note that this subdivision produces a conformally equivalent polyhedral surface to $\S$ with identity as the conformal equivalence, and therefore approximating the uniformization for this new polyhedral surface is equivalent to approximating the uniformization of the original polyhedral surface $\S$.

%Let us note that once Conjecture \ref{con:discretization_of_power_map} is proved no such subdivision will be necessary since for any vertex $v_i$, the sum of angles $\theta_i$ at $v_i$ is greater (in non-degenerate cases) twice any adjacent angle $\theta_{i,j}$ (of face $f_j$ touching $v_i$), and we will have that $\gamma=2\pi/\theta_i < 2\pi / 2\theta_{i,j} = \pi/\theta_{i,j}$.

%We will actually require slightly more from our association of charts at level $q=Q$: we ask that there exists a constant $R'>0$ such that for all faces $f_j\in\F^Q$ and their charts $\set{\Omega_i,\phi_i}$ the distance (using the induced piecewise Euclidean metric) from the triangle to the boundary of the charts' domain $\partial\Omega_i$ is at-least $R'$, that is, $$\min_{p\in f_j,p'\in\partial\Omega_i} d_\S(p,p')>R',$$ where $d_\S(p,p')$ is the geodesic distance on $\S$. Intuitively, we ask that each triangle will be of some bounded  distance (from below) to the boundary of its chart. For boundary vertices that are not corner vertices we use similar requirement with the difference that we measure the distance to $\partial\Omega_i\setminus \partial S$, that is, we do not care about the distance to the boundary edge. The reason is that we later use Schwarz reflection principle and the boundary edges won't be an obstacle.  \yl{what about corner charts?}

%%%%%%%%%%%%%%%%%%%%%%%%%%%%%%%%%%%%%%%%%%%%%%%%%%%%%%%%%%%%%%%%%%%%%%%%%%%%%%%%%%%%%%%%%
\section{Feasibility}
\label{s:feasibility}
We are now ready to prove our first result. Namely, that a subdivided triangulation $\S^q$, $q\geq 2$, can be mapped to $\T$ with a simplicial homeomorphism $\Phi^q\in \FK^{\S^q}$ where the discrete conformal distortion of the triangles is controlled.

% is globally bounded $\CD(\Phi^q)\leq K^*$, for some $K^*>0$ independent of $q$, and for arbitrary open set $\Omega\subset \S, \Omega\cap\V=\emptyset$ satisfies asymptotically $\CD_\Omega(\Phi^q)=1+\O(2^{-q})$, where $\O(2^{-q})$ is a quantity that its absolute value is not larger than some constant times $2^{-q}$.

\begin{thm}\label{thm:feasibility_result}
There exist simplicial maps $\Phi^q\in \FF^{\S^q}_{\K}$, for $\K=(K_1,...,K_{|\F^q|})$, $K_j=1+\O(2^{-\kappa_i q})$, where $\kappa_i=\min\set{\gamma_i,1}$ and face $f_j\in\F^q$ is associated with chart $i=i_j$. That is, $\Phi^q:\S\too\T$ are orientation preserving homeomorphisms that satisfy for every face $f_j\in\F^q$
$$\DCD(\Phi^q\mid_{f_j})=1+\O(2^{-\kappa_i q}).$$
In particular, $\Phi^q\in\FF^{\S^q}_{K_q}$, where $K_q=1+\O(2^{-\kappa q})$, and $\kappa=\min_i \set{\kappa_i}$.\\
%\item
Furthermore, for every face $f_j\in\F^q$, $\min_{k\in\Z}\abs{\arg\partial_{z}\A_j-\arg\brac{ (\Psi\circ\vphi_i^{-1})'(\wt{z})}+k2\pi}=\O\parr{2^{-q\kappa}}$, where $\A_j$ is $\Phi^q\vert_{f_j}$ in local coordinates, $\wt{z}$ is the centroid of the triangle $\Delta(z_1,z_2,z_3)$, $z_\ell=\vphi_i(v_{j_\ell})$, $\ell=1,2,3$, and $v_{j_\ell}$ are the vertices of the face $f_j$. \\
Lastly, for any fixed domain $\Omega\subset\S$ satisfying $\closure{\Omega}\footnote{We denote by $\closure{\Omega}$ the closure of the set $\Omega$.}\cap \V = \emptyset$,  $\DCD_\Omega(\Phi^q)=1+\O(2^{-q})$.
\end{thm}

We will need an auxiliary Lemma regarding approximation of conformal mappings with simplicial maps:
\begin{lem}\label{lem:sampling_conformal_with_triplet}
Let $f: \D(0,r)\to \D(0,1)$ be a conformal map such that $|f'(0)|\geq c_{der}>0$.  Let $z_\ell\in D(0,r)$, $\ell=1,2,3$, such that $z_1+z_2+z_3=0$, and the minimal angle of the triangle  $\Delta(z_1,z_2,z_3)$ is bounded from below, and denote $h=\max_\ell \abs{z_\ell}$. Then, the affine map $\A(z)=\alpha z+\beta \bbar{z} + \delta$, $\alpha,\beta,\delta\in\C$, defined uniquely by $\A(z_\ell)=f(z_\ell)$, $\ell=1,2,3$, satisfies:
\begin{enumerate}
\item $\abs{\alpha-f'(0)} = \O(h)$.
\item $\abs{\beta} = \O(h)$.
\item $\CD(\A)=1+\O\parr{\frac{h}{r(r-h)}}$, with the constant inside the $\O$-notation depending only upon the minimal angle of the triangle $\Delta(z_1,z_2,z_2)$ spanned by $z_1,z_2,z_3$.
\item $\min_{k\in \Z}\abs{\arg\alpha_j-\arg f'(0)+2\pi k}=\O(h)$. \end{enumerate}

\end{lem}
\begin{proof}
 We use the Taylor expansion (see for example,\cite{ahlfors1979complex} page 179) of $f$ developed around $z=0$:
\begin{equation}\label{e:taylor_expansion}
f(z)-f(0)=f'(0)z+\frac{z^2}{2\pi\i}\int_C \frac{f(\xi)\,d\xi}{\xi^2(\xi-z)},
\end{equation}
where $C=\partial \closure{\D(0,r)}$ is the circle of radius $r$ centered at the origin. We bound the reminder term for $f(z_\ell),\ell=1,2,3,$ as follows. For $z_\ell, \ell=1,2,3$: $$\abs{\frac{z_\ell^2}{2\pi\i}\int_C \frac{f(\xi)\,d\xi}{\xi^2(\xi-z_\ell)}}\leq \frac{h^2}{r(r-h)}.$$

Denote, for brevity $\eps=\frac{h^2}{r(r-h)}$. We can write three equations that characterize the affine map $\A(z)=\alpha z+ \beta \bbar{z} + \delta$ : $$\alpha z_\ell + \beta\bbar{z_\ell} + \delta = f(0) + f'(0)z_\ell + \eps_\ell  \ \ , \ell=1,2,3,$$
where $\abs{\eps_\ell}\leq \eps$. After rearranging:
$$ \parr{\alpha-f'(0)} z_\ell + \beta\bbar{z_\ell} + \parr{\delta-f(0)} = \eps_\ell\ \ , \ell=1,2,3.$$
Let us consider the matrix of the linear system:
\begin{equation}\nonumber
A=\begin{array}{cccc}

                   \begin{pmatrix}
                     z_1 & \bbar{z_1} & 1 \\
                     z_2 & \bbar{z_2} & 1 \\
                     z_3 & \bbar{z_3} & 1 \\
                   \end{pmatrix}.
\end{array}\end{equation}
Cramer's rule implies that we can bound $|\alpha-f'(0)|$ and $|\beta|$ by bounding $\abs{\det A_\ell / \det A}$, $\ell=1,2$, where $A_\ell$ is identical to $A$ except that we replace its $\ell^{th}$ column with the vector $(\eps_1,\eps_2,\eps_3)^t$. A direct calculation shows that
$$\det A = (z_2-z_1)\bbar{(z_3-z_1)}-(z_3-z_1)\bbar{(z_2-z_1)}=2\i\det  \begin{pmatrix}
                     \re(z_3-z_1) &  \im(z_3-z_1) \\
                     \re(z_2-z_1) &  \im(z_2-z_1) \\
                   \end{pmatrix}.$$
Therefore, $\abs{\det A}=4|\Delta|$,  where $\abs{\Delta}$ denotes the area of the triangle $\Delta$ spanned by $z_1,z_2,z_3$. The terms of the form $\abs{\det A_\ell}$ can be bounded by $$\abs{\det A_\ell} \leq 6\eps h, \quad \ell=1,2.$$ Combining the above we get $$|\alpha-f'(0)| \leq \frac{\abs{\det A_\ell}}{\abs{\det A}}\leq \frac{6\eps h}{4\abs{\Delta}}=\frac{3}{2}\frac{h^3}{r(r-h)\abs{\Delta}}=\O(h),$$ where in the last equality we used the fact that the minimal angle in triangle $\Delta$ is bounded from below and therefore $|\Delta|\geq c'h^2$. And similar bound holds for $|\beta|$.

Since $\abs{f'(0)}\geq c_{der}> 0$, there exists some constant depending on $c_{der}$ such that, up-to an integer multiplication of $2\pi$, the difference  $\arg f'(0)-\arg\alpha$ are of the same order as $\abs{\alpha-f'(0)}$, namely, $$\min_{k\in \mathbb{Z}}\abs{\arg\alpha-\arg f'(0)+2\pi k}=\O(h).$$

Finally,
$$\CD(\A)=\frac{|{\alpha}|+|{\beta}|  } {\abs{|{\alpha}|-|{\beta}|}}=\frac{\abs{\alpha}+\O(h)}{\abs{\alpha}+\O(h)}=1+\O(h),$$ where we used the fact that $\abs{f'(0)}\geq c_{der}>0$ to bound $\alpha = f'(0) + \O(h)$ away from zero.
The lemma is proved.
\end{proof}

%proof of theorem
\begin{proof}(of Theorem \ref{thm:feasibility_result})
We construct $\Phi^q$ by sampling the uniformization map $\Psi$, that is, for $v_i\in\V^q$, we set $\Phi^q(v_i)=\Psi(v_i)$, and extend $\Phi^q$ to the whole $\S^q$ by requiring linearity in each face.

Each face $f_j\in\F^q$, $q\geq Q$, is associated with some chart $\set{\Omega_i,\vphi_i}$, $i=i_j$, (used to define its conformal structure, see Section \ref{s:smooth_and_discrete_conformal_structures}). We denote, as above, by $v_{j_\ell}, \ell=1,2,3$ the (ordered) vertices of $f_j$, and $z_\ell=\vphi_i(v_{j_\ell}), \ell=1,2,3$ their image in the local coordinates.\\
We start with considering $f_j$ that are associated with \emph{interior charts}, namely, charts $\set{\Omega_i,\vphi_i}$ for $v_i\in \interior \S$. Lemma \ref{lem:epsilon_distance_to_charts_boundary} assures existence of a constant $\eps_0>0$ such that $\cup_{p\in f_j}\D(\vphi_i(p),\eps_0)\subset \Lambda_i$. Let us denote by $\wt{z}$ the centroid of the triangle $\Delta\parr{z_1,z_2,z_3}$. Let us note that the edges' length of $\S^q$ are asymptotically of order $2^{-q}$. In the local coordinates (i.e., $\Lambda_i$) the edge length goes to zero not slower than $2^{-q\kappa_i}$, where $\kappa_i=\min\set{\gamma_i,1}$ (i.e., depends upon the vertex angles $\theta_i$ and how different they are from $\Theta_i$ - angle deficit). Let us denote $h=\max_{\ell=1,2,3}|z_\ell-\wt{z}|$ the radius of disk around $\wt{z}$ containing $\Delta\parr{z_1,z_2,z_3}$, and the conformal map $\psi_i=\Psi \circ \vphi_i^{-1}$. The derivative of the conformal map, $\psi'_i$, can be bounded from below over $\closure{\Lambda_i}$, namely $\abs{\psi'_i(z)}\geq c_{der}>0$, $z\in \Lambda_i$ (since $\vphi_i$, and therefore $\psi_i$ can be extended to a neighborhood of $\closure{\Lambda_i}$). Since we have a finite number of charts and faces in $\F$ we can take $c_{der}$to be a uniform lower bound for all charts.

We now would like to use Lemma \ref{lem:sampling_conformal_with_triplet}. We assume w.l.o.g.~ that $\wt{z}=0$. We set $r=\eps_0$, and note that $\T \subset \closure{\D(0,1)}$. Also note that all the faces in every subdivision level $\S^q$ are similar to the original faces $\F$ of $\S$. Lemma \ref{lem:discretization_of_power_map} therefore implies that the triangle $\Delta=\Delta(z_1,z_2,z_3)$ have bounded angles from below. Lemma \ref{lem:sampling_conformal_with_triplet} now implies that the conformal distortion of the unique affine map $\A_j$ taking $\Delta$ to the triangle in $\T$ spanned by $\psi_i(z_\ell),\ell=1,2,3$ is bounded by $1+\O\parr{\frac{h}{r(r-h)}}=1+\O(2^{-\kappa_i q})$, since $h=\O(2^{-q \kappa_i})$. In addition this lemma indicates that $\abs{\arg\psi_i'(\wt{z})-\arg\alpha_j}=\O(2^{-\kappa_i q})$, where $\alpha_j=\partial_z\A_j$ (as-usual, the arguments are considered up-to addition of $2\pi k$).

 If we fix a domain $\Omega\subset \S$ such that $\closure{\Omega}\cap \V = \emptyset$ then for $f_j\subset\Omega$ the edge length of $\Delta(z_1,z_2,z_3)$ is asymptotically $2^{-q}$ and we achieve that the conformal distortion of $\A_j$ is bounded with $1+\O(2^{-q})$. Lastly, for this case, note that Lemma \ref{lem:sampling_conformal_with_triplet} also indicates that for sufficiently high $q$ the affine map $\A_j$ is orientation preserving and non-degenerate. Indeed since $\abs{\psi'_i(\wt{z})}\geq c_{der}>0$ one can use Lemma \ref{lem:sampling_conformal_with_triplet} to show that for sufficiently large $q$, $\det\A_j=|\alpha_j|^2-|\beta_j|^2>0$. \\

Now we move to faces $f_j$ associated with \emph{boundary charts}.
In this case we can use Schwarz reflection principle to extend $\psi_i$ to the union of $\Lambda_i$ and its reflection over its straight line boundary. Note that we can bound the derivative of the extension away from zero also in these cases (find a new $c_{der}>0$ that works also for boundary charts). Now we are again in the situation where we can apply Lemma \ref{lem:sampling_conformal_with_triplet}. (we can adjust the constant $\eps_0>0$ to also work for these charts.)\\

%\begin{floatingfigure}[r]{0.3\textwidth}
%\begin{wrapfigure}{0.3\textwidth}
  \begin{center}
    \includegraphics[width=0.7\textwidth]{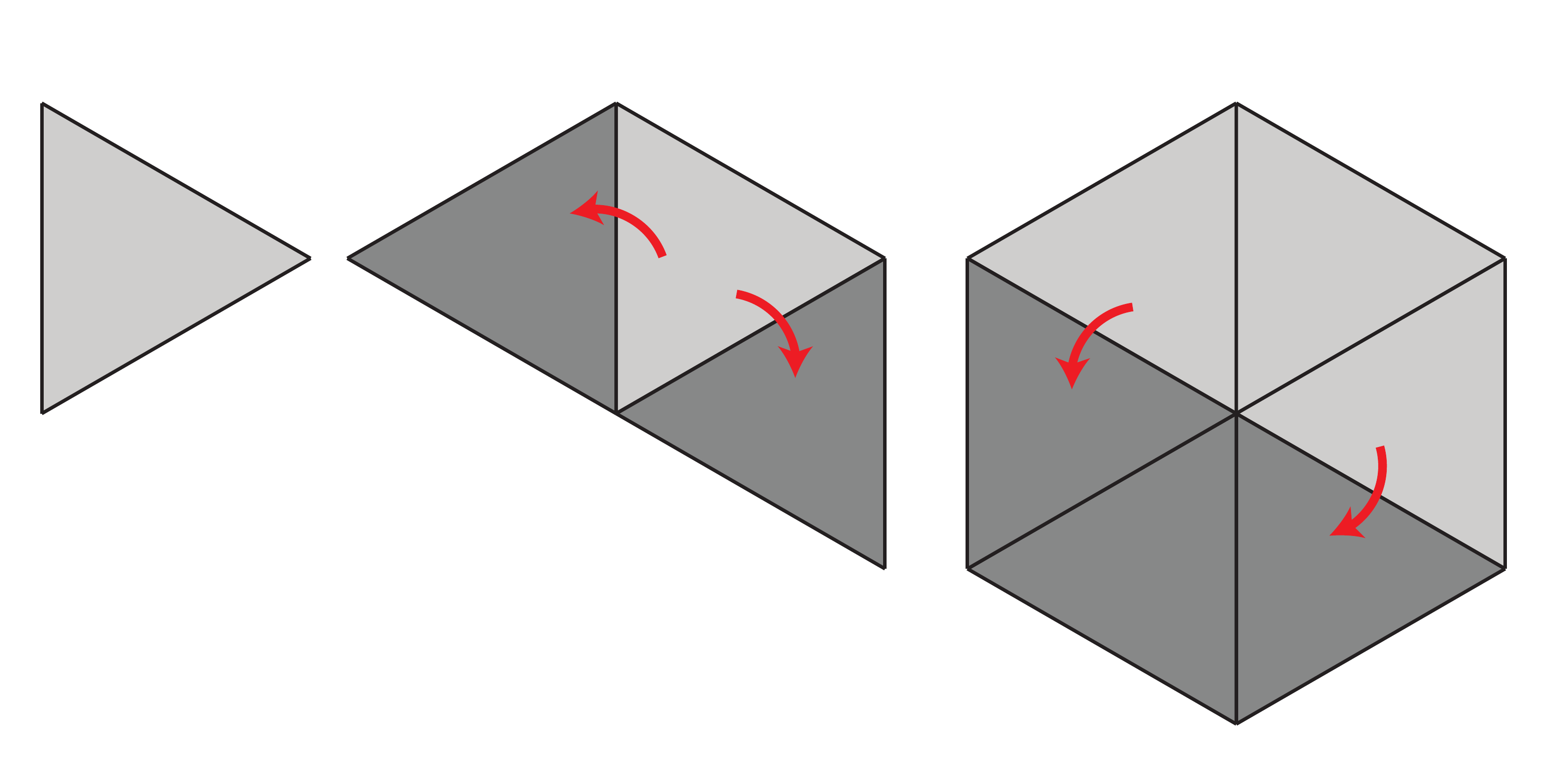}%\vspace{-0.2cm}
  \end{center}
  %\caption{The Projection approach.}
  %\label{fig:multi_connected}
%\end{wrapfigure}
%\end{floatingfigure}
The last case includes faces $f_j$ associated with \emph{corner charts} $\set{\Omega_i,\vphi_i}$, $i=1,2,3$. This case can also be dealt with the Schwarz reflection principle, as follows. By definition of our conformal structure, the uniformization map in the corner coordinate charts $\psi_i:\Psi\circ\vphi_i^{-1}:\Lambda_i\too\T$, $i=1,2,3$, takes the domain $\Lambda_i=\vphi_i(\Omega_i)$ to a neighborhood of one of the corners of the equilateral $\T$. Since the angle of the corner of $\Lambda_i$ equals $\frac{\pi}{3}$ which is also the angle of the equilateral, we can use Schwarz reflection principle to extend $\psi_i$ to a conformal map in a  neighborhood of $\Lambda_i$, as shown in the figure above. We first reflect w.r.t.~two each of the edges touching the corner, and then reflect w.r.t the new line boundary (similarly to other, non-corner, boundary charts). The extended map is analytic (and conformal) also at the corner point because it is continuous there and conformal in its neighborhood.\\

To finish the proof we need to show that $\Phi^q$ are homeomorphic mappings from $\S$ to $\T$. Note that $\Phi^q\mid_{f_j}$ is an affine map that is composed of two affine maps. The first, denoted by $\B_j$, takes $f_j$ in $\S^q$ to the coordinate chart's triangle $\Delta=\Delta(z_1,z_2,z_3)$, and the second, $\A_j$, maps $\Delta$ to $\Delta(\Psi(v_{j_1}),\Psi(v_{j_2}),\Psi(v_{j_3}))$. The first affine map is orientation preserving and non-degenerate (and with bounded conformal distortion) for sufficiently large $q$ due to Lemma \ref{lem:discretization_of_power_map}. The second affine map $\A_j$ was proven above (by showing that the determinant of its Jacobian is positive) to be orientation preserving and non-degenerate as well for all faces and for sufficiently large $q$. Therefore, $\Phi^q\mid_{f_j}$ is orientation preserving non-degenerate as a composition of two orientation-preserving and non-degenerate affine maps. Since $\Phi^q$ also maps the boundary of $\S$ bijectively onto the boundary of $\T$ (by definition of $\FF^{\S^q})$), it is a global bijection (see \cite{Lipman:2012:BDM:2185520.2185604} for a detailed proof). Since it is simplicial, its inverse is also continuous, a fact which makes it a homeomorphism.
\end{proof}

A comments is in order. One can build the charts based on an already subdivided version of the surface $\S^q$, $q>0$, and achieve better estimates of the conformal distortion of the maps $\Phi^q$. For example, if one takes the charts defined by $\S^1$, the vertices $v_i\in\V^1\setminus\V$ are all flat, i.e. $\theta_i=2\pi$, and therefore triangles associated with the corresponding charts will have conformal distortion of order $1+\O(2^{-q})$.

\newpage
%%%%%%%%%%%%%%%%%%%%%%%%%%%%%%%%%%%%%%%%%%%%%%%%%%%%%%%%%%%%%%%%%%%%%%%%%%%%%%%%%%%%%%%%%
\section{Approximation}
\label{s:approximation}

Our next step is showing that if we are able to put our hands on a series of simplicial maps with discrete conformal distortion that goes uniformly to one, then this series has to uniformly converge to the uniformization map $\Psi$. Note that this result applies to \emph{any} such series, not only the ones
built in Theorem \ref{thm:feasibility_result}.
%Furthermore, up to a conformal change of coordinates the rate of convergence is linear in $\abs{\DCD(\Phi^q)-1}=\O(2^{-q\gamma })$, where $\kappa=\min\set{\kappa_i}$, and as defined before $\kappa_i=\min\set{\gamma_i,1}$ (see Section \ref{s:smooth_and_discrete_conformal_structures}, and Section \ref{s:feasibility}).
\begin{thm}\label{thm:approximation}
Let $\Psi:\S\too\T$ be the uniformization map of the polyhedral surface $\S$.  Let $\set{\Phi^q}_{q\geq Q} \subset \FF^{\S^q}_{K_q}$, for some series $K_q=1+\O(2^{-q\kappa})$, where $\kappa>0$ is some constant.\\
Then, $\Phi^q\circ\Psi^{-1}:\T\too\T$ converges uniformly in any compact subset of $\interior{\T}$ to the identity map $I_d:\T\too\T$.
\end{thm}
\begin{proof}
Define $g^q=\Phi^q\circ\Psi^{-1}:\T\too\T$.\\ We start by showing that $g^q$ is $K$-quasiconformal with some global distortion bound $K\geq 1$.\\ First, since $\Psi$ is an orientation-preserving homeomorphism, $\Psi^{-1}$ is orientation-preserving homeomorphism. By assumption $\Phi^q$ is an orientation-preserving  homeomorphism and so $g^q$ is an orientation-preserving homeomorphism. Let us observe that $g^q_j=g^q\mid_{\Psi\parr{\interior{f_j}}}$ can be written as  $$g^q_j=\Phi^q~\circ \vphi_i^{-1}\circ\underbrace{\vphi_i\circ\Psi^{-1}}_{(*)}\mid_{\Psi\parr{\interior{f_j}}},$$ where $i=i_j$ is the chart associated with the face $f_j\in\F^q$. The part that is marked with $(*)$ is a conformal map over $\Psi\parr{\interior{f_j}}$ (by the definition of the uniformization map). Next, recall that $\vphi_i$ consists of a composition of a rigid maps and power map $h(z)=z^{\gamma_i}$, which is conformal except at the origin. Therefore, $\vphi^{-1}_i$ is conformal over the domain $\vphi_i\circ\Psi^{-1}\parr{\Psi\parr{\interior{f_j}}}=\vphi_i\parr{\interior{f_j}}$. Lastly, the $\Phi^q$ part is defined over the set $\vphi_i^{-1}\circ\vphi_i\circ\Psi^{-1}\parr{\Psi\parr{\interior{f_j}}}=\interior{f_j}$ and can be written as a composition of two affine maps (as before) the first, denoted by $\B_j$, takes $f_j$ in $\S^q$ to the coordinate chart triangle $\Delta=\Delta\parr{\vphi_i(v_{j_1}),\vphi_i(v_{j_2}),\vphi_i(v_{j_3})}$,where $v_{j_\ell}$, $\ell=1,2,3$ are the (ordered) vertices of face $f_j$. The second affine map, $\A_j$, maps $\Delta$ to $\Delta(\Psi(v_{j_1}),\Psi(v_{j_2}),\Psi(v_{j_3}))$. $\B_j$ is $K'$-quasiconformal, for some $K'\geq 1$ independent of $q$ by Lemma \ref{lem:discretization_of_power_map}. $\A_j$ is $K''$-quasiconformal, again independently of $q$, due to the assumption that $\Phi^q \in \FF^{\S^q}_{K_q}$, $K_q=1+\O(2^{-q\kappa})$. Then, $g^q_j$ is $K=K'K''$-quasiconformal, for all faces $f_j\in\F^q$ and all $q$. Denote the domain $$D=\cup_{f_j\in\F^q}\Psi\parr{\interior{f_j}}\subset\interior{\T}.$$ The set $\T\setminus D$, which consists of the image under $\Psi$ of the edges $\E^q$ of $\S^q$ is a union of analytic arcs. By extension theorem for quasiconformal mappings (see e.g., Theorem 8.2, page 42, and Theorem 8.3, page 45, in \cite{lehto1973quasiconformal}) the conformal distortion of $g^q$ over the whole $\T$ is $K$ again. Since bounded $K$-quasiconformal maps form a normal family there exists a subsequence $g^{q_k}$ that converges locally uniformly to a $K$-quasiconformal map or a constant (see Theorem 3.1.3, page 49 in \cite{astala2009elliptic}). Denote the limit map by $g$. Since $g^q$ fixes the corners of $\T$, $g$ is $K$-quasiconformal. Next we show that $g$ is a conformal map. It is enough to prove that $g$ is 1-quasiconformal (see e.g., Theorem 5.1, page 28 in \cite{lehto1973quasiconformal}). For that end, take an arbitrary domain $O\subset \T$ such that $\closure{O}\cap\Psi\parr{\V}=\emptyset$. Now, we look at $g^{q_k,O}_j=g^{q_k}\mid_{O\cap\Psi(\interior{f_j})}$. Similar to above, it consists of a composition of conformal maps ($\vphi_i\circ\Psi^{-1}$ and the inverse of the power map $h(z)=z^{\gamma_i}$), and an affine map that can be represented as composition of two affine maps $\B_j$ and $\A_j$. The conformal distortion of $\A_j$ equals $\CD(\A_j)=1+\O(2^{-q\kappa})$ by assumption. The affine map $\B_j$ is defined by sampling the power map $h(z)=z^{\gamma_i}$, $i=i_j$ at the corners of $f_j$ (the origin is placed at vertex $v_i\in\V$) and extending linearly. Since, by the intersection with the domain $\O$, we sample $h(z)$ at some bounded (from below) distance from the origin (vertex $v_i$), Lemma \ref{lem:sampling_conformal_with_triplet} indicates that $\CD(\B_j)=1+\O(2^{-q})$. This implies that the conformal distortion of $g_j^{q_k,O}$ is $1+\O(2^{-q\kappa})$. As before, extension theorems for quasiconformal mappings implies that $g^{q_k}\mid_O$ is $1+\O(2^{-q\kappa})$-quasiconformal. Taking arbitrary large $Q$ and observing the series tail $\Phi^q$ for $q\geq Q$ implies that $g\mid_O$ is $1$-quasiconformal (for relevant convergence theorem, see for example Theorem 5.2 in \cite{lehto1973quasiconformal} page 29). So we get that for arbitrary domain $O\subset \T$ such that $\closure{O}\cap\Psi\parr{\V}=\emptyset$,
$g\mid_O$ is conformal. This means that $g$ is conformal in $\interior{\T}\setminus\Psi(\V)$. However $\Psi(\V)$ is a set of discrete points and therefore standard extension theorems of conformal mappings imply $g$ is conformal in the whole $\interior{\T}$. Since $g$ fixes the corners of the equilateral it has to be the identity map since it is the only conformal bijection from $\T\too\T$ that fixes the three corners.

Since every infinite subsequence of $\set{g^q}$ has a subsequence converging locally uniformly to the identity, $g^q$ has to converge to the identity map locally uniformly.

\end{proof}

%%%%%%%%%%%%%%%%%%%%%%%%%%%%%%%%%%%%%%%%%%%%%%%%%%%%%%%%%%%%%%%%%%%%%%%%%%%%%%%%%%%%%%%%%
\section{Convex Spaces of Simplicial Quasiconformal mappings}
\label{s:algorithm}
In this section we develop the convex spaces of simplicial quasiconformal mappings, prove the main theorem of this paper (Theorem \ref{thm:main}), and develop two algorithms for approximating the uniformization map $\Psi:\S\too\T$.

Let us recap what we know so far:
\begin{enumerate}
\item
\textbf{Uniformization:} there exists a unique conformal homeomorphism of the polyhedral surface to the equilateral $\Psi:\S\too\T$ taking three prescribed boundary vertices $v_1,v_2,v_3$ to the corners $t_1,t_2,t_3$ of the equilateral $\T$.

\item
\textbf{Feasibility:} for sufficiently large $q$, $\FF^{\S^q}_{1+\O(2^{-q\kappa})}$ is not empty, where $\kappa=\min\set{\kappa_i}$, $\kappa_i=\min\set{1,\gamma_i}$. Moreover, $\FF^{\S^q}_\K$ is not empty, where $\K=(K_1,...,K_{|\F^q|})$, and $K_j=1+\O(2^{-q\kappa_i})$ for $f_j\in\F^q$ associated with the chart $i=i_j$.

\item \textbf{Approximation:} any series $\set{\Phi_q}_{q\geq Q}$,  $\Phi_q \in \FF^{\S^q}_{1+\O(2^{-q\kappa})}$, where $\kappa>0$ is a constant, converge to the uniformization map $\Psi$ on compact subsets of $\interior{\T}$.
\end{enumerate}

These observations motivate the meta-algorithm shown in Algorithm \ref{alg:approx_of_uniformization_map} for approximating the uniformization map $\Psi$.

\begin{algorithm}[h!]%\small
% \dontprintsemicolon
  \KwIn{Polyhedral surface $\S=(\V,\E,\F)$, with three positively-oriented boundary points $v_1,v_2,v_3\in\V$\\ \quad \quad \quad \
        Subdivision level $q\geq 2>0$ \\ \quad \quad \quad \
        Constant $0<c<1$}
  \KwOut{Quasiconformal simplicial map $\Phi^q$}
%\SetLine
\BlankLine
\BlankLine
\tcp{Set the charts}
\ForAll{$v_i\in\V$}
{
Calculate $\theta_i,\gamma_i,\kappa_i$\;
}
\BlankLine
\tcp{Subdivide}
Subdivide $\S$ $q$ times to get $\S^q=\parr{\V^q,\E^q,\F^q}$\;
\BlankLine
\tcp{compute discrete conformal structure}
\ForAll{$f_j\in\F^q$}
{
Associate $f_j$ with a chart $i=i_j$\;
Choose a coordinate system on $f_j$ such that $v_i$ is placed at the origin\;
Map the vertices $v_{j_\ell},\ell=1,2,3$ of $f_j$ by $z\mapsto z^{\gamma_i}$, and set as $\Delta_j$\;
$K_j\leftarrow 1+2^{-c q\kappa_{i}}$\;
}
$\K\leftarrow \parr{K_1,K_2,..,K_{|\F^q|}}$\;
\BlankLine
\tcp{Approximate}
Find a map $\Phi^q\in\FF^{\S^q}_{\K}$\;
\BlankLine
Return $\Phi^q$\;
\caption{Approximation of the uniformization map.}
\label{alg:approx_of_uniformization_map}%
\end{algorithm}

The correctness of this algorithm can be explained as follows: since $0<c<1$, the series $2^{-c q \kappa_{i_j}}$ approaches zero slower than $\O(2^{-q\kappa_{i_j}})$, Theorem \ref{thm:feasibility_result} guarantees that for sufficiently large $q$ the set $\FF^{\S^q}_\K$ with $K_j=1+2^{-c q\kappa_{i_j}}$ is not empty. Theorem \ref{thm:approximation} then implies that $\Phi^q$ will converge to the uniformization map locally uniformly.

At this point we are left with the problem of finding an element in the non-empty set $\FF^{\S^q}_\K$. This is a non-trivial task since $\FF^{\S^q}_\K$ is a non-convex set and hence finding an element in this space is a non-convex problem. Nevertheless, as we explain next, we can identify a collection of convex subsets of $\FF^{\S^q}_\K$ the union of which spans $\FF^{\S^q}_\K$.

The space $\FF^{\S^q}_\K$ contains orientation-preserving homeomorphic simplicial maps $\Phi^q$ such that $\DCD(\Phi^q\mid_{f_j})=\CD(\A_j)\leq K_j$, where $\A_j$ is defined by mapping the triangle
$\Delta_j=\Delta_j\parr{z_1,z_2,z_3}$, $z_\ell=\vphi_i(v_{j_\ell})$, $\ell=1,2,3$ (i.e., the triangle formed by the vertices of $f_j$
written in local coordinates) to the triangle
$\Delta(\Phi^q(v_{j_1}),\Phi^q(v_{j_2}),\Phi^q(v_{j_3}))$ (i.e., the image of $f_j$ under $\Phi^q$).
Next, we write down characterizing equations for $\FF^{\S^q}_\K$.

We start by examining a single face $f_j\in \F^q$. We will work with its image in the local coordinates as given by its associated chart $\set{\Omega_i,\vphi_i}$, $i=i_j$. An affine map $\A=\A_j$ mapping $\Delta_j$ to the complex plane can be written (using standard complex notation) as $$\A(z)=\alpha z + \beta \bbar{z} + \delta, \quad \alpha,\beta,\delta \in\C.$$
In what follows we will only consider the linear part of $\A$, namely $\alpha z+\beta\bbar{z}$. We will abuse notation and denote $\A(z)=\alpha z +\beta\bbar{z}$, the translation does not affect conformal distortion nor the orientation and therefore it will be safe to ignore it for now. We ask that $\A$ is orientation-preserving and $\CD(\A)\leq K$, for some arbitrary but fixed $K\geq 1$. The conformal distortion of $\A$ can be calculated using the formula (see e.g., \cite{ahlfors1966lectures}, page 11)
\begin{equation}\label{e:CD_of_A}
    \CD(\A) = \frac{|\alpha|+|\beta|}{|\alpha|-|\beta|} \leq K.
\end{equation}
And the Jacobian of $\A$,
\begin{equation}\label{e:J_of_A}
    \J(\A) = |\alpha|^2-|\beta|^2.
\end{equation}

$\A$ is an orientation preserving homeomorphism if and only if $\J(\A) > 0$. Therefore, necessary and sufficient conditions for the affine map $\A$ to be $K$-quasiconformal are \begin{eqnarray}
  \CD(\A) & \leq & K \label{e:D_of_A_leq_K}\\
  \J (\A) & > & 0 \label{e:J_of_A_ge_0}
\end{eqnarray}

We will write the above conditions in a more convenient form. First, using (\ref{e:J_of_A}), (\ref{e:J_of_A_ge_0}) is equivalent to
 \begin{equation}\label{e:alpha_ge_beta}
    |\alpha| > |\beta|.
 \end{equation}
 Second, using (\ref{e:CD_of_A}), (\ref{e:D_of_A_leq_K}) can be rearranged, taking into account that (\ref{e:alpha_ge_beta}) implies $\alpha \ne 0$ we get
\begin{equation}\label{e:alpha_leq_dil_beta}
    |\beta| \leq \frac{K-1}{K+1}|\alpha|, \quad \alpha\ne 0
\end{equation}

Let us denote the collection of orientation preserving planar linear transformations with conformal distortion $K$ by $\FK$. The collection $\FK$ can now be thought of as a subset of $\C\times\C$ where a pair $(\alpha,\beta)\in\C\times\C$ represents an element (without translation) $\A\in\FK$, namely, $\A(z)=\alpha z+ \beta \bbar{z}$.
%Note that $\FK$ is scale invariant, that is, if $\A(z)=\alpha z+\beta\bbar{z} \in \FK$, then also $\wh{\A}=r\alpha z + r\beta \bbar{z} \in \FK$, for $r\in \Real\setminus\set{0}$.

 In this parameterization we characterize the maximal convex subspaces of $\FK$. As we prove next there is one parameter family of maximal convex subspaces $\FKtau\subset\FK\subset \C\times \C$ defined by:

%\begin{floatingfigure}[r]{0.25\textwidth}
%%\begin{wrapfigure}{r}{0.3\textwidth}
%  \begin{center}
%    \includegraphics[width=0.25\textwidth]{figures/FKtet.pdf}%\vspace{-1.5cm}
%  \end{center}
%  %\caption{The Projection approach.}
%  %\label{fig:multi_connected}
%%\end{wrapfigure}
%\end{floatingfigure}

%\begin{eqnarray}
% % \nonumber to remove numbering (before each equation)
%   |\beta| & \leq & r \label{e:FFKtet_1}\\
%   \re \parr{e^{-\i \tau} \alpha} & \geq & \frac{K+1}{K-1}r \label{e:FFKtet_2}\\
%   r & > & 0 \label{e:FFKtet_3}
% \end{eqnarray}
%
%
\begin{eqnarray}
   |\beta| & \leq & \frac{K-1}{K+1}\re \parr{e^{-\i \tau} \alpha},\quad \alpha\ne0  \label{e:FFKtet_1}
 \end{eqnarray}

From this definition we see that $\FKtau$ are convex subsets of $\FK$ (this can verified directly by taking convex combinations of elements in $\FKtau$). These convex spaces span $\FF_K$, $$\FK=\bigcup_{\tau\in[0,2\pi)}\FKtau.$$  Lastly, $\FKtau$ are maximal convex subsets of $\FK$. We say that  $U\subset \FK$ is a \emph{maximal convex subset} if every convex set $V \subset \FK$ that contains it $U\subset V$ has to be equal to it, $U=V$.

\begin{lem}\label{lem:optimality_of_FKtet}
$\FKtau$, $\tau\in[0,2\pi)$ are maximal convex subsets of $\FK$.
\end{lem}
\begin{proof}
Let $U \subset \FK$ be a convex subset such that $\FKtau \subsetneqq U$. Let $(\alpha^*,\beta^*)\in U\setminus \FKtau$. By definition (\ref{e:FFKtet_1}) $\FKtau$ contains the linear map $(e^{\i \tau}\frac{K+1}{K-1}|\beta^*|,\beta^*)$. Now since $U$ is convex it contains the convex combinations $$\parr{\alpha(\lambda),\beta(\lambda)}:=\parr{\parr{1-\lambda}e^{\i \tau}\frac{K+1}{K-1}|\beta^*| + \lambda \alpha^*, \beta^*}\in U.$$
Since $(\alpha^*,\beta^*)\notin\FK$ we have
that $$|\beta^*| \frac{K+1}{K-1} >  \re \parr{e^{-\i \tau} \alpha^*}.$$
Using this inequality, \begin{eqnarray*}
\abs{\alpha(\lambda)}^2 &<& \parr{\frac{K+1}{K-1}}^2\abs{\beta^*}^2\parr{1-\lambda^2}+\abs{\alpha^*}\lambda^2 \\
& = & \parr{\frac{K+1}{K-1}}^2\abs{\beta^*}^2 + \lambda^2\parr{\abs{\alpha^*} - \parr{\frac{K+1}{K-1}}^2\abs{\beta^*}^2}.
\end{eqnarray*}
Hence, for sufficiently small $\lambda$ we get that $$\abs{\alpha(\lambda)}<\frac{K+1}{K-1}\abs{\beta^*}=\frac{K+1}{K-1}\abs{\beta(\lambda)},$$
which is a contradiction with the fact that $U\subset\FK$.\end{proof}

The subset $\FKtau\subset\FK$ contains affine maps $\A(z)=\alpha z+\beta \bbar{z} +\delta \in \FK$ with distortion $\CD(\A)=k\leq K$ that in addition satisfy
\begin{equation}\label{e:angle_restriction}
\abs{\arg \alpha - \tau}\leq\cos^{-1}\parr{\frac{K+1}{K-1}\frac{k-1}{k+1}}.
\end{equation}
As usual, the argument in the l.h.s.~should be understood up-to addition of $k2\pi$, $k\in \Z$.
To see this, let $\A(z)=\alpha z+\beta\bbar{z}+\delta$ have conformal distortion $\CD(\A)=k\leq K$, and assume it satisfies equation (\ref{e:angle_restriction}) then we need to show that $\A\in\FK$. Indeed, we have
\begin{eqnarray*}
\re(e^{-\i\tau}\alpha)&=&\abs{\alpha}\cos\parr{\arg(\alpha)-\tau} \\
&\geq& \abs{\alpha}\frac{K+1}{K-1}\frac{k-1}{k+1} \\
&=& \abs{\beta}\frac{K+1}{K-1},
\end{eqnarray*}
where in the last equality we used the fact that $\CD(\A)=k$. These inequalities also show that the angle in the r.h.s. of (\ref{e:angle_restriction}) cannot be enlarged while $\A\in\FKtau$. This shows that (\ref{e:angle_restriction}) is necessary and sufficient condition for affine map $\A$ with distortion $\CD(\A)=k$ to belong to $\FKtau$. We proved
\begin{lem}\label{lem:angle_restriction}
The convex space $\FKtau$ contains all affine maps $\A$ with conformal distortion $k=\CD(\A)\leq K$ that their rotation (i.e., $\arg \partial_z\A$) satisfies equation (\ref{e:angle_restriction}).
\end{lem}

A corollary of this observation is that there does not exist a convex subset $U\subset\FK$ such that it contains all affine maps $\A$ with conformal distortion $k=\CD(\A)\leq K$ with larger angle limit than described in (\ref{e:angle_restriction}). Indeed, if there was such a convex space it would be a superset of $\FKtau$ which will contradict the maximality of $\FKtau$ shown in Lemma \ref{lem:optimality_of_FKtet}. We summarize:
\begin{cor}\label{cor:no_possible_larger_angle}
There is no convex subset of the $K$-quasiconformal simplicial mapping space, $U\subset \FK$, that contains all the affine maps $\A$ with conformal distortion $\CD(\A)\leq K$ and larger rotation angle limit than $\FKtau$ for any $\tau\in [0,2\pi)$.
\end{cor}

Intuitively, the maximal convex space $\FKtau$ restricts the rotation angle of its member affine transformations $\A$ (i.e., $\arg \alpha$) to be the maximal possible around a prescribed rotation by $\tau$ radians.
For perfect similarities $\A(z)=\alpha z$ the rotation angle range is $(\tau-\frac{\pi}{2},\tau+\frac{\pi}{2})$, while for maps with conformal distortion $k=K$ only rotation by exactly  $\tau$ is allowed. Affine maps with intermediate distortion values $1 <  k < K$ will be allowed rotations in between these two extremal cases, as expressed in equation (\ref{e:angle_restriction}).

% \yl{the following is equivalent to the proposition in the sig submission}
% \begin{lem}
% $\FMKtet$ contains affine mappings with conformal distortion $k\leq K$ %under the following restrictions: XXX
% \end{lem}

We now move to the general triangulations $\S^q$, and their corresponding spaces $\FF^{\S^q}_{\K}$ of homeomorphic simplicial maps $\Phi^q:\S^q\too\T$ where $\DCD\parr{\Phi^q\mid_{f_j}}\leq K_j$, $\K=\parr{K_1,..,K_{|\F^q|}}$. Our goal is to find an element in $\FF^{\S^q}_\K$. We start by formulating a set of equations that characterize $\FF^{\S^q}_\K$ exactly. We derive necessary conditions for $\Phi^q\in\FF^{\S^q}_\K$, and later show that they are also sufficient. Let $\Phi^q\in\FF^{\S^q}_\K$, then, over each face $f_j\in\F^q$, $\Phi\mid_{f_j}$ is an affine map, in the local coordinates it has the form
$$\A_j(z) = \alpha_j z + \beta_j \bbar{z} + \delta_j.$$
Also set $u_i=\Phi^q(v_i)\in\C$. By definition we have $\CD(\A_j)=K_j\leq K$. That is, $\A_j$ should satisfy the equation
\begin{equation}\label{e:alpha_j_leq_dil_beta_j}
    |\beta_j| \leq \frac{K_j-1}{K_j+1}|\alpha_j|, \quad \alpha_j\ne 0
\end{equation}
Furthermore, since $\Phi^q$ is continuous, each affine map $\A_j$ should map the vertices of its face $f_j$ to $u_i$. Namely, if we denote by $j_\ell$, $\ell=1,2,3,$ the indices of the vertices of face $f_j$, and $z_{j_\ell}=\vphi_i(v_{j_\ell})$ their complex coordinates, then
\begin{equation}\label{e:compatibility_eqs}
\A_j(z_{j_\ell})=\alpha_j z_{j_\ell} + \beta_j~ \bbar{z_{j_\ell}}+\delta_j = u_{j_\ell},\quad \ell=1,2,3.
\end{equation}
Note that these are homogeneous linear equations.

Another set of necessary conditions is related to the boundary conditions. We want the boundary $\partial\S^q$ to be mapped to the boundary of the equilateral $\partial \T$. This can be achieved by first forcing the corner vertices $v_1,v_2,v_3$ to be mapped to $\T$'s corners $t_1,t_2,t_3$. For each face $f_j$ that contains one of the corner vertices $v_\ell$, $\ell=1,2,3$ we add an equation of the form
\begin{equation}\label{e:corner_vertices_to_triangle_corners}
\A_j(z_{\ell})=\alpha_j z_{\ell} + \beta_j~ \bbar{z_{\ell}}+\delta_j = t_\ell.
\end{equation}
Second, the rest of the boundary vertices $v_i\in \partial \S^q$ are mapped to the edges of $\T$. This can be achieved by ordering the vertices that are supposed to be mapped to one edge of $\T$, say $(t_1,t_2)$, w.r.t. the positive orientation:~$v_{i_1}=v_1,v_{i_2},...,v_{i_{n-1}},v_{i_n}=v_2$. Now introduce new real variables $\lambda_2,...,\lambda_{n-1}$ and set the equations
\begin{eqnarray} \label{e:boundary_edges_1}
0&\leq&\lambda_2 \leq \lambda_3 \leq ... \leq \lambda_{n-1} \leq 1\\\label{e:boundary_edges_2}
u_{i_\ell}&=&(1-\lambda_\ell)t_1+\lambda_\ell t_2,\quad \ell=2,..,n-1.
\end{eqnarray}
We can use $\leq$ instead of $<$ in equation (\ref{e:boundary_edges_1}) since (\ref{e:alpha_j_leq_dil_beta_j}) implies strict inequality automatically.

To recap, we derived necessary conditions for $\Phi^q$ to belong to $\FF^{\S^q}_\K$. Let us show that they are also sufficient, that is, every piecewise affine map that satisfies equations (\ref{e:alpha_j_leq_dil_beta_j}),(\ref{e:compatibility_eqs}),(\ref{e:corner_vertices_to_triangle_corners}),(\ref{e:boundary_edges_1}),(\ref{e:boundary_edges_2})
belongs to $\FF^{\S^q}_\K$.
\begin{prop}\label{prop:eqs_for_Phi_are_sufficient}
Let $\Phi^q$ be a piecewise affine map over $\S^q$ that satisfies equations (\ref{e:alpha_j_leq_dil_beta_j}),(\ref{e:compatibility_eqs}),(\ref{e:corner_vertices_to_triangle_corners}),(\ref{e:boundary_edges_1}),(\ref{e:boundary_edges_2}).
Then, $\Phi^q\in\FF^{\S^q}_\K$.
\end{prop}
\begin{proof}
From equations (\ref{e:compatibility_eqs}) we get that $\Phi^q$ is continuous. Equations (\ref{e:alpha_j_leq_dil_beta_j}) imply that over each face $f_j$, $\Phi$ is an orientation-preserving affine map with conformal distortion bounded by $K_j$. This implies that $\Phi:\S^q\too\T$ is a bijection (see \cite{Lipman:2012:BDM:2185520.2185604} for a proof). Since $\Phi^q$ is an orientation preserving piecewise-affine continuous bijection it is an orientation preserving homeomorphism. Therefore $\Phi^q\in\FF^{\S^q}_\K$.
\end{proof}

Going back to Algorithm \ref{alg:approx_of_uniformization_map}, we can find an element in $\FF^{\S^q}_\K$ by simply looking for variables $\set{u_i},\set{\alpha_j,\beta_j,\delta_j},\set{\lambda_\ell}$ that satisfies equations (\ref{e:alpha_j_leq_dil_beta_j}),(\ref{e:compatibility_eqs}),(\ref{e:corner_vertices_to_triangle_corners}),(\ref{e:boundary_edges_1}),(\ref{e:boundary_edges_2}).
In fact, equation (\ref{e:compatibility_eqs}) can be used to express the variables $\set{\alpha_j,\beta_j,\delta_j}$ as linear combinations of the variables $\set{u_i}$ and therefore eliminate these variables from the equations. All the mentioned equations are convex equations except equations (\ref{e:alpha_j_leq_dil_beta_j}). Since the bounded conformal distortion equations (\ref{e:alpha_j_leq_dil_beta_j}) are not convex, finding a feasible solution $\Phi^q$ is not a convex problem. However, we can replace equation (\ref{e:alpha_j_leq_dil_beta_j}) with the convex conditions (\ref{e:FFKtet_1}) (describing the maximal convex subset)  and achieve a convex problem. That is, for each face $f_j\in\F^q$, we need to assign an angle $\tau_j\in[0,2\pi)$, and replace (\ref{e:alpha_j_leq_dil_beta_j}) with the equation
\begin{eqnarray}
   |\beta_j| & \leq & \frac{K_j-1}{K_j+1}\re \parr{e^{-\i \tau_j} \alpha_j},\quad \alpha_j\ne 0.   \label{e:FFKtet_j_1}
 \end{eqnarray}
Let us denote by $\FF^{\S^q}_{\K,\Tau}$, $\Tau=\parr{\tau_1,\tau_2,...,\tau_{|\F^q|}}\in[0,2\pi)^{|\F^q|}$ the convex space defined by equations (\ref{e:FFKtet_j_1}),(\ref{e:compatibility_eqs}),
(\ref{e:corner_vertices_to_triangle_corners}),(\ref{e:boundary_edges_1}),(\ref{e:boundary_edges_2}). By construction and Proposition \ref{prop:eqs_for_Phi_are_sufficient}, $\FF^{\S^q}_{\K,\Tau}\subset \FF^{\S^q}_\K$. In Algorithm \ref{alg:approx_of_uniformization_map} we set $K_j=1+2^{-cq\kappa_{i_j}}$, where $i_j$ is the index of the chart associated with face $f_j$. We proved in Theorem \ref{thm:feasibility_result} that there exists a simplicial map, let us denote it $\wh{\Phi}^q$, such that the discrete conformal distortion $\DCD(\wh{\Phi}^q\mid_{f_j})$, that is $\CD(\wh{\A}_j)$ equals $1+\O(2^{-q\kappa_{i_j}})$. Denote $\wh{\A}_j(z)=\wh{\alpha}_jz+\wh{\beta}_j\bbar{z}+\wh{\delta}$, and $\wh{\tau}_j=\arg \wh{\alpha}_j$. Equation (\ref{e:angle_restriction}) now implies that $\FF^{\S^q}_{\K,\Tau}$ will contain $\wh{\Phi}^q$ if  $$\abs{\wh{\tau}_j - \tau_j}\leq\cos^{-1}\parr{\frac{2+2^{-cq\kappa_{i_j}}}{2^{-cq\kappa_{i_j}}}\frac{\O(2^{-q\kappa_{i_j}})}{2+\O(2^{-q\kappa_{i_j}})}} ,$$  for all $f_j\in\F^q$, and the r.h.s. converge to $\frac{\pi}{2}$ as $q\too \infty$ (remember that the constant inside the $\O$-notation is set per chart and there is a finite number of charts). Therefore, for any $\eps>0$, there exists sufficiently large $q$ such that $\wh{\Phi}^q\in\FF^{\S^q}_{\K,\Tau}$, as long as $\abs{\tau_j-\wh{\tau}_j}<\frac{\pi}{2}-\eps$. (as before, the arguments should be understood up-to addition of $k2\pi$, $k\in \Z$) We have proved:
\begin{lem}\label{lem:feasiblity_in_convex_Ftau}
Let $\wh{\Phi}_j:\S^q\too \T$ be the simplicial map, the existence of which is set by Theorem \ref{thm:feasibility_result}. Denote its restriction to face $f_j$ in local chart coordinates by $\wh{\A}_j(z)=\wh{\alpha}_jz+\wh{\beta}_j\bbar{z}+\wh{\delta}$, and set $\wh{\tau}_j=\arg \wh{\alpha}_j$. \\ For arbitrary $\eps>0$ there exists sufficiently large $q$ such that if $\Tau=(\tau_1,..,\tau_{|\F^q|})$ satisfy $$\max_{f_j\in\F^q}\min_{k\in\Z}\abs{\tau_j-\wh{\tau}_j+k2\pi}<\frac{\pi}{2}-\eps,$$
then, $$\wh{\Phi}^q\in\FF^{\S^q}_{\K,\Tau}$$
\end{lem}

We are now in a position to prove the main theorem of this paper:\\

\textbf{Theorem \ref{thm:main}}
\textit{
Let $\Psi:\S\too\T$ be the uniformization map of a disk-type polyhedral surface to the equilateral $\T$ taking three prescribed boundary vertices of $\S$ to the corners of $\T$. Let $\S^q$ be the $q^{th}$-level subdivided version of $\S$. \\
If, for an arbitrary but fixed $\eps>0$, the argument of $\Psi'$ is known up to an error of $\pm(\frac{\pi}{2}-\eps)$, then one can construct a series of non-empty convex spaces $U^q=\set{\Phi^q}\subset \FF^{\S^q}$ of simplicial maps of $\S^q$ such that:
\begin{enumerate}
\item
Every map $\Phi^q \in U^q$ is $K$-quasiconformal (QC) homeomorphism that maps $\S$ onto $\T$, with some constant $K$ independent of $q$.
\item
Every series $\set{\Phi^q}_{q\geq Q}$, where $\Phi^q\in U^q$, converges locally uniformly to the uniformization map $\Psi$. That is, $\Phi^q\circ\Psi^{-1}:\T\too\T$ converges uniformly in any compact subset of $\interior{\T}$ to the identity map $I_d:\T\too\T$.
\end{enumerate}}
\begin{proof}
Let $\eps>0$ be arbitrary but fixed. Let $\wh{\Phi}^q:\S^q\too\T$ be the simplicial map, the existence of which is assured by Theorem \ref{thm:feasibility_result}. Denote, as before, $\wh{\A}_j(z)=\wh{\alpha}_jz+\wh{\beta}_j\bbar{z}+\wh{\delta}$, and $\wh{\tau}_j=\arg \wh{\alpha}_j$, $z_\ell=\vphi_i(v_{j_\ell})$, $i=i_j$, $\ell=1,2,3$, and $\Delta=\Delta(z_1,z_2,z_3)$ the face $f_j$ in the local coordinate chart, and the centroid $\wt{z}=\frac{1}{3}\parr{z_1+z_2+z_3}$.
Let (by Theorem  \ref{thm:feasibility_result} again)  $Q>0$ be sufficiently large such that for all $q\geq Q$, $f_j\in\F^q$,
\begin{equation}\label{e:for_main_thm_1}
\min_{k\in \Z}\abs{\wh{\tau}_j-\arg \brac{\parr{\Psi\circ\vphi^{-1}_i}'(\wt{z})}+k2\pi}<\frac{\eps}{2}.
\end{equation}
Now we set $U^q=\FF^{\S^q}_{\K,\Tau}$, where $\K=\parr{K_1,..,K_{|\F^q|}}$, $K_j=1+2^{-cq\kappa_{i}}$, for some constant $1>c>0$, $\Tau=\parr{\tau_1,..,\tau_{|\F^q|}}$. $\FF^{\S^q}_{\K,\Tau}$ is defined by equations (\ref{e:FFKtet_j_1}),(\ref{e:compatibility_eqs}),
(\ref{e:corner_vertices_to_triangle_corners}),(\ref{e:boundary_edges_1}),(\ref{e:boundary_edges_2}), and Proposition \ref{prop:eqs_for_Phi_are_sufficient} implies that  $\FF^{\S^q}_{\K,\Tau}\subset \FF^{\S^q}_{K}$, in particular, it contains only $K$-quasiconformal homeomorphic simplicial maps from $\S\too\T$, for some sufficiently large but constant $K$.  By assumption we know the argument of $\Psi'$ up-to $\frac{\pi}{2}-\eps$ and so we take some $\tau_j$ such that
\begin{equation}\label{e:for_main_thm_2}
\min_{k\in \Z}\abs{\tau_j-\arg \brac{ \parr{\Psi\circ\vphi^{-1}_i}'(\wt{z})}+k2\pi}<\frac{\pi}{2}-\eps.
\end{equation}
By triangle inequality, (\ref{e:for_main_thm_1}) and (\ref{e:for_main_thm_2}) imply
$$\min_{k\in\Z}\abs{\tau_j-\wh{\tau}_j+k2\pi}<\frac{\pi}{2}-\frac{\eps}{2}.$$
Now we want to use Lemma \ref{lem:feasiblity_in_convex_Ftau}.  Set $Q>0$ sufficiently large as required by Lemma \ref{lem:feasiblity_in_convex_Ftau} (we could have set $Q>0$ sufficiently large a-priori) and we get $\wh{\Phi}^q\in\FF^{\S^q}_{\K,\Tau}$, and so it is not empty, for all $q\geq Q$. Now take an arbitrary series $\set{\Phi_q}_{q\geq Q}$, $\Phi^q\in \FF^{\S^q}_{\K,\Tau}$, where $K_j=1+\O(2^{-cq{\kappa}_i})=1+\O(2^{-q\wt{\kappa}})$, $\wt{\kappa}=c\min\set{\kappa_i}>0$. The convergence of the series $\Phi^q$ to the uniformization map is implied by Theorem \ref{thm:approximation}.
\end{proof}

An immediate consequence of Theorem \ref{thm:main} is that if we have approximations $\tau_j$ of the argument of $\Psi'$ up-to an error of $\pm\frac{\pi}{2}$ then we can simply solve the convex feasibility problem $\Phi^q\in\FF^{\S^q}_{\K,\Tau}$ to achieve an approximation of the uniformization map. In particular, this can be achieved by solving the following convex minimization problem:
\begin{eqnarray}\nonumber
\min &\eps& \\ \label{e:convex_feasibility_given_Tau}
\mathrm{s.t.} & &\\ \nonumber
|\beta_j| & \leq & \frac{K_j-1}{K_j+1}\re \parr{e^{-\i \tau_j} \alpha_j} +\eps,\\ \nonumber
\textrm{and} && eqs.~(\ref{e:compatibility_eqs}),(\ref{e:corner_vertices_to_triangle_corners}),(\ref{e:boundary_edges_1}),(\ref{e:boundary_edges_2}).
\end{eqnarray}

Next we suggest two algorithms on how to set $\Tau$ without any prior knowledge: the first is exhaustive, guaranteed to find the feasible $\Tau$ but requires solving (\ref{e:convex_feasibility_given_Tau}) very large number of times (exponential in the number of faces of $\S$). This algorithm is guaranteed to work, but its computational complexity makes it impractical for applications. The second algorithm is an iterative one that start from arbitrary $\Tau^0$ and construct a series of $\Tau^1,\Tau^2,...$ until finds a feasible space $\FF^{\S^q}_{\K,\Tau}$. The second algorithm is much more efficient than the first one, however lacking a proof that a feasible space is always found (for sufficiently high $q$). Nevertheless, once found a solution, all the theoretical guarantees apply and the convergence and approximation results apply. In practice this algorithm works well as we demonstrate at the end of the paper.

\subsection{First Algorithm (exhaustive)}

We can exhaustively search for a feasible $\Tau$ by testing for each $\tau_j$ three angles, say $\tau_j\in\set{0,\frac{2\pi}{3},\frac{4\pi}{3}}$. This way we are guaranteed to find an assignment $\Tau=(\tau_1,..,\tau_{|\F^q|})$ satisfying the condition of Theorem \ref{thm:main}. However, testing all such assignments requires solving (\ref{e:convex_feasibility_given_Tau}) for $3^{|\F^q|}$ convex spaces $\FF^{\S^q}_{\K,\Tau}$, which is computationally daunting. However, it is clear that the argument of the derivative of the uniformization map cannot change arbitrarily in small areas, and indeed, as expected, we will show that the number of convex spaces needed to be searched is at-most of order $3^{\abs{\F}}$, which is still high but at-least independent of the subdivision level $q$.

It is possible to prove this by developing an approximation argument for conformal maps, using the approximation Lemma \ref{lem:sampling_conformal_with_triplet} (observing how the argument of the derivative of the affine maps approximate the arguments of the derivative of the conformal map). However, we will pick a different route which we believe also gives an interesting intuition regarding quasiconformal simplicial maps. Intuitively, the claim is that the argument change in simplicial quasiconformal mappings is linearly (and simply) bounded by the conformal distortion bound. Denote by $d_{\SU}\parr{e^{\i \tau},e^{\i\tau'}}=\min_{k\in \Natural }\abs{\tau-\tau'+2\pi k}$ the distance on the circle $\SU$, then,

%\begin{floatingfigure}[t]{0.2\textwidth}
\begin{figure}[t]
%  \begin{center}
    \includegraphics[width=0.3\textwidth]{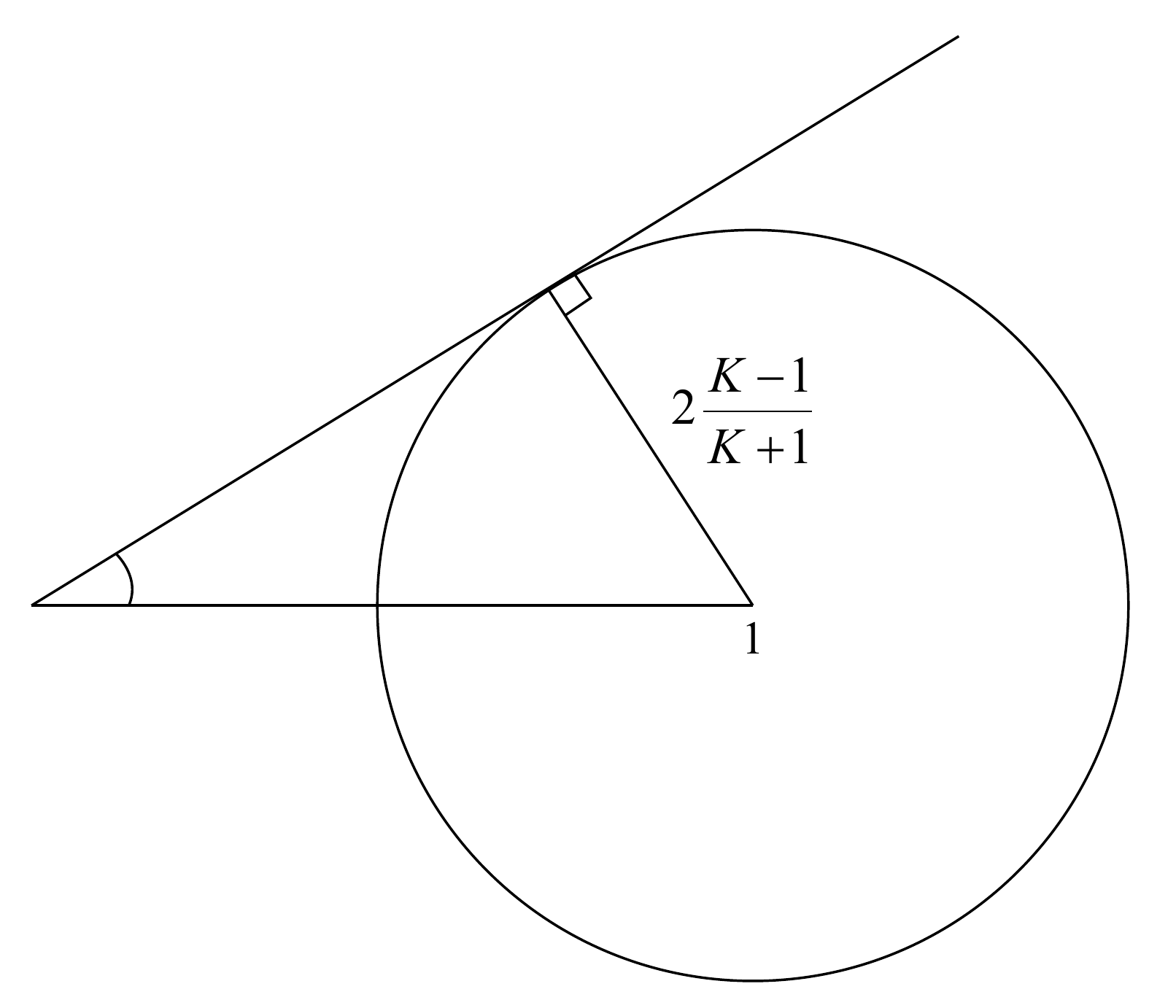}%\vspace{-0.2cm}
%  \end{center}
  \caption{Illustration for the proof of Lemma \ref{lem:arg_bound_for_KQC_maps}.}
  \label{fig:for_proof_bound_arg_KQC}
\end{figure}
%\end{floatingfigure}

\begin{lem}\label{lem:arg_bound_for_KQC_maps}
Let $f_j,f_{j'}$ be two planar faces sharing an edge. Let $\A_j,\A_{j'}$ be two planar affine maps, coinciding on the common edge and have $K$ conformal distortion bound, that is $\CD(\A_j)\leq K$, and $\CD(\A_{j'})\leq K$.
Then $$\abs{d_{\S^1}\parr{\frac{\partial_z\A_j}{\abs{\partial_z\A_j}},\frac{\partial_z\A_{j'}}{\abs{\partial_z\A_{j'}}}}}\leq \sin^{-1}\parr{2\frac{K-1}{K+1}}.$$
\end{lem}
\begin{proof}
Let $e_k$ denote the common edge of the faces $f_j,f_{j'}$, and let $e_k=(v_i\,v_{i'})$, $v_i,v_{i'}$ its two vertices. By continuity $$\A_j(v_{i'}-v_i)=\A_{j'}(v_{i'}-v_i).$$
Plugging the expressions $\A_j(z)=\alpha_j z+\beta_j\bbar{z}+\delta_j$, $\A_j(z)=\alpha_{j'} z+\beta_{j'}\bbar{z}+\delta_{j'}$ for $\A_j,\A_{j'}$ we get $$\alpha_j(v_{i'}-v_i)+\beta_j \bbar{(v_{i'}-v_i)} = \alpha_{j'}(v_{i'}-v_i)+\beta_{j'} \bbar{(v_{i'}-v_i)},$$
and so
$$\parr{\alpha_j-\alpha_{j'}}(v_{i'}-v_i) = \parr{\beta_{j'}-\beta_j}\bbar{(v_{i'}-v_i)}.$$
Taking the absolute value and dividing by $\abs{v_{i'}-v_i}$ we get $$\abs{\alpha_j-\alpha_{j'}}= \abs{\beta_{j'}-\beta_j}.$$ Since $\CD(\A_j)\leq K$, we have that $|\beta_j|\leq \frac{K-1}{K+1}|\alpha_j|$, and similarly for $\A_{j'}$. Let us assume w.l.o.g.~that $|\alpha_j|\geq |\alpha_{j'}|.$ Therefore,
$$\abs{\alpha_j-\alpha_{j'}}= \abs{\beta_{j'}-\beta_j} \leq |\beta_j|+|\beta_{j'}|\leq \frac{K-1}{K+1}\parr{|\alpha_j|+|\alpha_{j'}|}\leq 2\frac{K-1}{K+1}|\alpha_j|.$$
% \begin{floatingfigure}[r]{0.4\textwidth}
% %\begin{wrapfigure}{l}{0.3\textwidth}
%   \begin{center}\vspace{-1.0cm}\hspace{2cm}
%     \includegraphics[width=0.2\textwidth]{figures/for_proof_of_bound_arg_KQC.pdf}%\vspace{1cm}
%   \end{center}
%   %\caption{The Projection approach.}
%   %\label{fig:multi_connected}
% %\end{wrapfigure}
% \end{floatingfigure}

Dividing both sides by $|\alpha_j|$, we get
$$\abs{1-\frac{\alpha_{j'}}{\alpha_j}}\leq 2\frac{K-1}{K+1}.$$
It is not hard to check now (see Figure \ref{fig:for_proof_bound_arg_KQC}) that $$\abs{d_{\S^1}\parr{\frac{\alpha_j}{\abs{\alpha_{j}}},\frac{\alpha_{j'}}{\abs{\alpha_{j'}}}}} = \abs{\arg\frac{\alpha_{j'}}{\alpha_j}}\leq\sin^{-1}\parr{2\frac{K-1}{K+1}}$$
\end{proof}

Lemma \ref{lem:arg_bound_for_KQC_maps} indicates that for $\wh{\Phi}^q\in\FF^{\S^q}_{\K}$, $K_j=1+\O(2^{-q\kappa_i})$,  a pair of adjacent faces $f_j,f_{j'}\subset\Omega_i$ associated to the same chart $\set{\Omega_i,\vphi_i}$, $$d_{\SU}\parr{\frac{\wh{\alpha}_j}{\abs{\wh{\alpha}_j}},\frac{\wh{\alpha}_{j'}}{\abs{\wh{\alpha}_{j'}}}}=\O\parr{2^{-q\kappa_i}},$$
where we used a Taylor expansion of $\sin^{-1}(\cdot)$ around zero. This is still not enough since to traverse from one face to another in the $q^{th}$ subdivision level inside $\Omega_i$ we will need $\O(2^q)$ edge crossings, and hence we only get that the difference between different $\arg\wh{\alpha}_j$ for $f_j\subset \Omega_j$ is only bounded by $\O(2^{q(1-\kappa_i)})$ and since $\kappa_i$ can be strictly smaller than one, this is still not a constant. We therefore need to refine our argument. We will show that the argument change in a chart $\Omega_i$ can be bounded independently of the subdivision level $q$.

% \begin{thm}\label{thm:argument_bounded_change}
% Let $\wh{\Phi}^q:\S^q\too \T$ be the simplicial map the existence of which is set by Theorem \ref{thm:feasibility_result}. Let $f_j,f_{j'}\in\F^q$ be two faces associated with chart $\set{\Omega_i,\vphi_i}$. Denote by $\wh{\A}_j(z)=\wh{\alpha}_jz+\wh{\beta}_j\bbar{z}+\wh{\delta}$ the local affine map of $\wh{\Phi}^q\mid_{f_j}$, and similarly for $\wh{\A}_{j'}$.
% Then, $$d_{S^1}\parr{\frac{\wh{\alpha}_j}{\abs{\wh{\alpha}_j}},\frac{\wh{\alpha}_{j'}}{\abs{\wh{\alpha}_{j'}}}}=\O\parr{1}.$$
% \end{thm}
%\begin{}
To do that we note that in Theorem \ref{thm:feasibility_result} the estimate $\DCD(\Phi^q\mid_{f_j})=\CD(\A_j)=1+\O(2^{-q\kappa_{i}})$ (remember that $\kappa_i=\min\set{1,\gamma_i}$) can be refined by replacing it with $$\DCD(\Phi^q\mid_{f_j})=\CD(\A_j)=1+\O\parr{h_j},$$
where $h_j=\max_{\ell=1,2,3}\abs{z_\ell-\wt{z}}$, $\wt{z}=\frac{1}{2}\parr{z_1+z_2+z_3}$ is the centroid of the triangle $\Delta(z_1,z_2,z_3)$, and $z_\ell=\vphi_i(v_{j_\ell})$ are the images in the local coordinates of the vertices of the face $f_j$, and the constant in the $\O$-notation is independent of which face $f_j$ we choose in chart $\set{\Omega_i,\vphi_i}$. Now, take two faces $f_j,f_{j'}\subset \Omega_i$ in chart $i=i_j$. W.l.o.g. we can assume they both belong to one $f\in\F$, that is $f_j,f_{j'}\subset f$, otherwise we can use the following argument few times but not more than half the valence of the vertex $v_i$ (defining the chart $i$). We can also assume that $f_j$ is the face touching the vertex $v_i$, otherwise we can bound the argument change by the sum of bounds when traversing from $f_j$ to the corner face (touching $v_i$) and from the corner face to $f_{j'}$.

%\begin{floatingfigure}[r]{0.25\textwidth}
\begin{wrapfigure}{r}{0.25\textwidth}\vspace{-0.5cm}
  \begin{center}
    \includegraphics[width=0.25\textwidth]{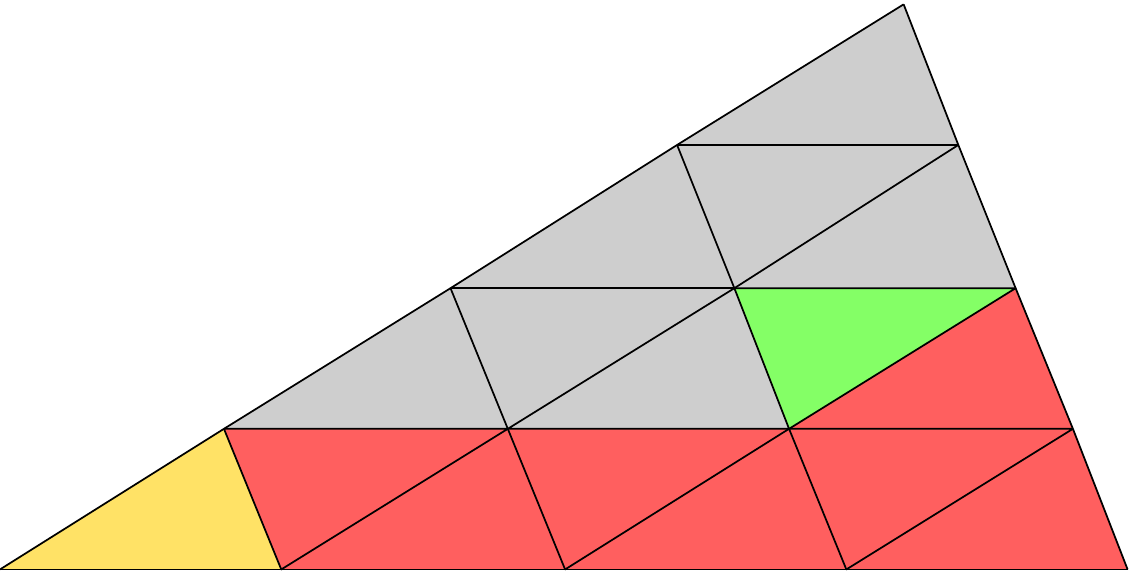}\vspace{-0.2cm}
  \end{center}
  %\caption{The Projection approach.}
  %\label{fig:multi_connected}
\end{wrapfigure}
%\end{floatingfigure}

Starting from (the corner) $f_j$ we traverse to any $f_{j'}$ in the manner depicted in the inset figure (red triangles represent the path connecting the face $f_j$, in yellow, and the face $f_{j'}$, in green). Denote the faces we traverse by $f_j\too f_{j_1} \too f_{j_2} \too ... \too f_{j'}$. We would now like to bound $\O(h_j)+\O(h_{j_1})+\O(h_{j_2})+...+\O(h_{j'})$. Since, as mentioned above, the constant in the $\O$-notation is independent of the choice of face in the chart we actually need to bound $h_j+h_{j_1}+h_{j_2}+...+h_{j'}$. This will be done by using the binomial formula.

%We will also use the same assumption of that lemma. Note that since we already made these assumption once, there is no harm assuming it again, but it seems possible to make the following arguments more general.
\begin{floatingfigure}[r]{0.25\textwidth}
%\begin{wrapfigure}{r}{0.25\textwidth}\vspace{-0.0cm}
  \begin{center}
    \includegraphics[width=0.25\textwidth]{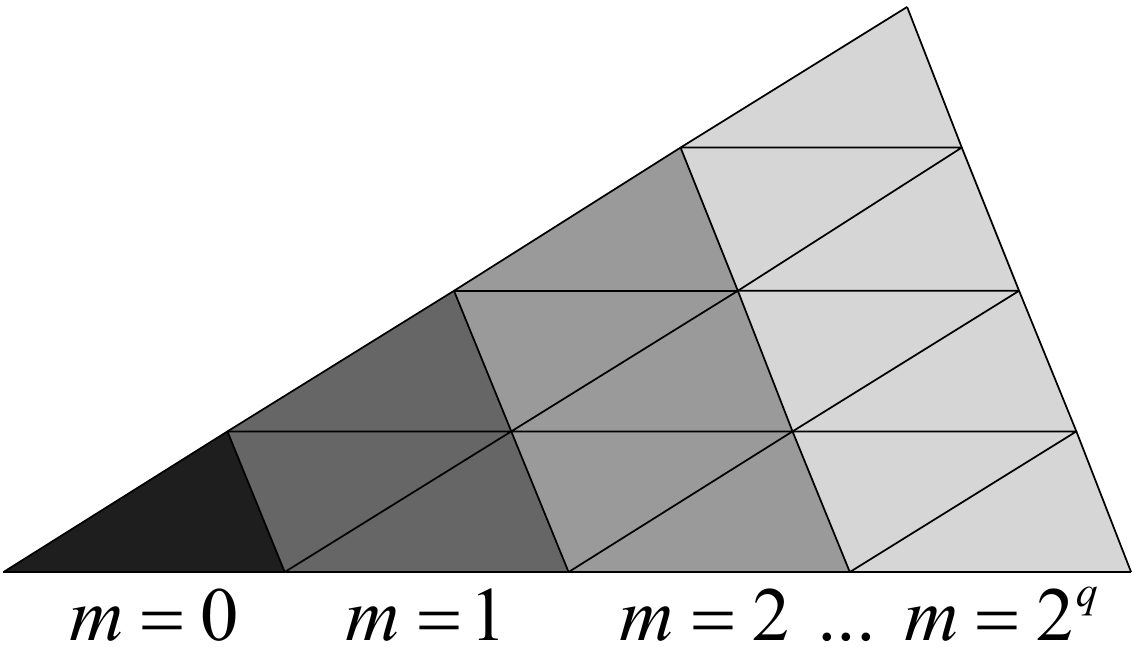}\vspace{-0.2cm}
  \end{center}
  %\caption{The Projection approach.}
  %\label{fig:multi_connected}
%\end{wrapfigure}
\end{floatingfigure}
The first thing we do is to bound $h$ for a general face $f_{j''}\subset f$. We will then use this bound to bound the sum of $h$'s for our path. To bound $h=\max_\ell|z_\ell-\wt{z}|$ it is enough to bound the length of one of the edges of the triangle $\Delta(z_1,z_2,z_3)$, where $z_\ell=\vphi_i(v_{j''_\ell})$, and $v_{j''_\ell}$,$\ell=1,2,3$ are the vertices of the face $f_{j''}$. Let us denote by $m$ the generation of $f_{j''}$ counted from the corner, as shown in the inset.
\begin{floatingfigure}[r]{0.25\textwidth}
%\begin{wrapfigure}{r}{0.25\textwidth}\vspace{-0.5cm}
  \begin{center}
    \includegraphics[width=0.25\textwidth]{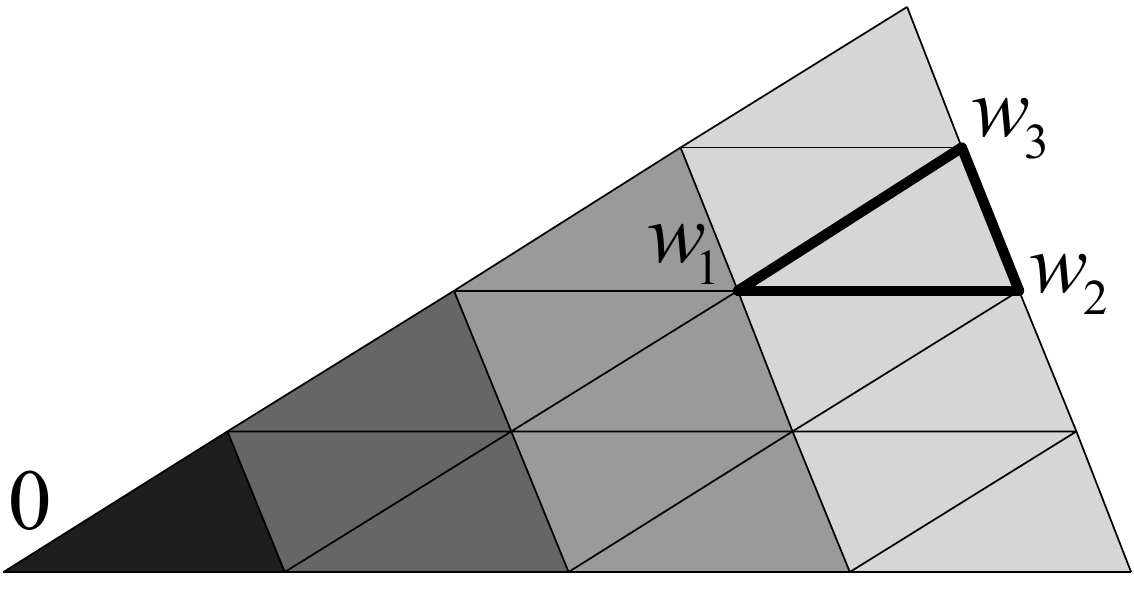}\vspace{-0.2cm}
  \end{center}
  %\caption{The Projection approach.}
  %\label{fig:multi_connected}
%\end{wrapfigure}
\end{floatingfigure}

%We repeat the assumption of Lemma \ref{lem:discretization_of_power_map} that $f$ is isosceles, supported by $\xi,\eta$ ($|\xi|=|\eta|$) and the enclosed angle $\theta$ satisfy $\theta <60.4^\circ$.
%\textcolor{red}{WHERE DO WE USE IT?}

Denote by $w_1,w_2,w_3$ the corners of $f_{j''}$. $\Delta(w_1,w_2,w_3)$ is similar to $f$, and w.l.o.g. we assume $w_1$ is the vertex corresponding to $f$'s corner under the similiarity (see inset figure for illustration). Then the length of one of the edges of $\Delta(z_1,z_2,z_3)$ is
\begin{eqnarray}\nonumber
\abs{w_1^{\gamma_i}-w_3^{\gamma_i}}&=& \abs{w_1^{\gamma_i}+\parr{w_1+\parr{w_3-w_1}}^{\gamma_i}} \\ \label{e:bounding_edge_length_1} &=& \abs{w_1}^{\gamma_i}\abs{1-\parr{1+\parr{\frac{w_3-w_1}{w_1}}}^{\gamma_i}}\\ \nonumber
&=&\abs{w_1}^{\gamma_i}\abs{1-\parr{1+\mu}^{\gamma_i}},
\end{eqnarray}
where we set $\mu=\frac{w_3-w_1}{w_1}$. Now using the binomial expansion $$(1+\mu)^{\gamma_i}=1+\gamma_i \mu+\parr{\begin{array}{c}
\gamma_i \\
2 \\
\end{array}} \mu^2+\parr{\begin{array}{c}
\gamma_i \\
3 \\
\end{array}} \mu^3+...=1+\gamma_i \mu +R,$$
and the reminder term $R$ can be bounded by noting that w.l.o.g. $\gamma_i<1$ (otherwise the edge length is $\O(2^{-q})$ and there is nothing to prove) and therefore by induction $\abs{\parr{\begin{array}{c}
\gamma_i \\
k \\
\end{array}} }\leq 1$, for $k\geq 2$, then,
$$\abs{R}\leq\abs{\mu}^2+\abs{\mu}^3+...=\frac{|\mu|^2}{1-|\mu|}.$$
Plugging back in (\ref{e:bounding_edge_length_1}) we get
\begin{eqnarray}\nonumber
\abs{w_1^{\gamma_i}-w_3^{\gamma_i}}&=& \abs{w_1}^{\gamma_i}\abs{1-\parr{1+\gamma_i\mu+R}}\\\nonumber
&\leq& \abs{w_1}^{\gamma_i}\parr{\gamma_i\abs{\mu}+\frac{|\mu|^2}{1-|\mu|}}.
\end{eqnarray}
It is not hard to see that there exists constants $c_1,c_2>0$ such that  $\abs{w_1}\leq c_1 m2^{-q}$, and $\abs{\mu}\leq c_2 m^{-1}$, and therefore
\begin{eqnarray}
\abs{w_1^{\gamma_i}-w_3^{\gamma_i}}&\leq& c~ 2^{-q\gamma_i} ~m^{\gamma_i-1}, \label{e:bounding_edge_length_2}
\end{eqnarray}
where $c>0$ is some constant. Now we use (\ref{e:bounding_edge_length_2}) to bound the sum of edge lengths in our path described above,
\begin{eqnarray*}
h_j+h_{j_1}+h_{j_2}+...+h_{j'}&\leq& \brac{ c2^{-q\gamma_i}\sum_{m=0}^{2^q}m^{\gamma_i-1}} + mc2^{-q\gamma_i}2^{q(\gamma_i-1)} \\ &\leq& c2^{-q\gamma_i}\parr{ \O\parr{2^{q\gamma_i}} + 2^{q\gamma_i}}\\&=&\O(1),
\end{eqnarray*}
where we used the fact that $m\leq 2^q$, and that $\sum_{m=0}^M m^{\alpha}=\O(M^{\alpha+1})$.
%\end{proof}

This means that we can consider only $\Tau=\parr{\tau_1,..,\tau_{|\F^q|}}$ that are defined on a subdivision level of, say, $q'= \lceil\log_4 q\rceil$. That is, for subdivision levels $q=2,3,4,5,6,...$, we take the assignment vector $\Tau$ to be piecewise constant on faces of $\S^{q'}$, $q'= 1,1,1,2,2,...$ (respectively). Since $\abs{\F^{q'}}=\O\parr{4^{q'}}=\O\parr{q}$, the number of assignments $\Tau$ we need to check to find a feasible map is of order $\O \parr{ 3^{\abs{\F^{q'}}}}=\O\parr{3^q}$. The arguments abobe imply that for sufficiently large $q$ there exists some assignment constant over faces of level $q'$ that will lead to a feasible convex problem (\ref{e:convex_feasibility_given_Tau}).

A comment is in order: as the argument of conformal maps is harmonic, it is very likely that the argument change will be slow even \emph{between} neighboring charts in $\S$ (not even the subdivided versions). Therefore,  it is possible to reduce the complexity of the algorithm for many polyhedral surfaces $\S$ by defining assignments by spreading small number of point seeds over $\S$, dividing $\S$ into sub-areas of ``constant'' argument and transporting the argument using the chart. Obviously, the way the arguments will be transported will change the transported values, but nevertheless this type of argument will result in an algorithm with complexity connected only to the geometry of $\S$ and its uniformization map to $\T$ (and not to its flat ``tessellation'').

\subsection{Second (iterative) algorithm}
Here we suggest a practical and simple algorithm for finding $\Phi^q\in \FF^{\S^q}_{\K}$. In a nutshell, this algorithm iteratively adjusts $\Tau$ based on the result of the previous iteration. Although we do not have a proof that this algorithm will always find a feasible solution (we will leave this issue to future work), it seems to do well in practice, and in the case it does find a feasible solution it is guaranteed to be an approximation of the uniformization map, and has all the theoretical properties of the exhaustive algorithm.

The idea of the algorithm is very simple: start with some arbitrary angle assignment $\Tau^0=(\tau^0_1,...,\tau^0_{|\F^q|})$ and solve the convex feasibility problem $\Phi^q\in\FF^{\S^q}_{\K,\Tau^0}$, where $\K=(K_1,..,K_{|\F^q|})$, and $K_j=1+2^{-cq\kappa_i}$. This is done by solving the cone programming problem (\ref{e:convex_feasibility_given_Tau}). If $\eps\leq0$ a feasible solution was found and we are done. Otherwise set $\tau^1_j=\arg\alpha_j$, where $\alpha_j$ is, as before, $\partial_z\A_j$. The outline of the algorithm is provided in Algorithm \ref{alg:iterative_approx_of_uniformization_map}.

\begin{algorithm}[t]%\small
% \dontprintsemicolon
  \KwIn{Polyhedral surface $\S=(\V,\E,\F)$, with three distinct boundary points $v_1,v_2,v_3\in\V$\\ \quad \quad \quad \
        Subdivision level $q\geq 2>0$ \\ \quad \quad \quad \
        Constants $0<c<1$}
  \KwOut{Quasiconformal simplicial map $\Phi^q$}
%\SetLine
\BlankLine
\BlankLine

\tcp{Set the charts}
\ForAll{$v_i\in\V$}
{
Calculate $\theta_i,\gamma_i,\kappa_i$\;
}
\BlankLine
\tcp{Subdivide}
Subdivide $\S$ $q$ times to get $\S^q=\parr{\V^q,\E^q,\F^q}$\;
\BlankLine
\tcp{compute discrete conformal structure}
\ForAll{$f_j\in\F^q$}
{
Associate $f_j$ with a chart $i=i_j$\;
Choose a coordinate system on $f_j$ such that $v_i$ is placed at the origin\;
Map the vertices $v_{j_\ell},\ell=1,2,3$ of $f_j$ by $z\mapsto z^{\gamma_i}$, and set as $\Delta_j$\;
$K_j\leftarrow 1+2^{-c q\kappa_{i}}$\;
}
$\K\leftarrow \parr{K_1,K_2,..,K_{|\F^q|}}$\;
\BlankLine

\tcp{In the following:~find a map $\Phi^q\in\FF^{\S^q}_{\K}$}
Initiate $\Tau$ arbitrarily: $\tau_j \in [0,2\pi)$\;
\While{$\eps$ is still decreasing}
{Solve (\ref{e:convex_feasibility_given_Tau})\;
Update $\Tau$: $\tau_j \leftarrow \arg \alpha_j$\;
}
\If{$\eps<0$}{Return $\Phi^q$\;}\Else{Return ``Did not find a feasible element in $\FF^{\S^q}_{\K}$.''\;}
\caption{Approximation of the uniformization map (iterative algorithm).}
\label{alg:iterative_approx_of_uniformization_map}%
\end{algorithm}

We have implemented this algorithm in MATLAB and used it to approximate uniformization of polyhedral surfaces.
One useful variation of Algorithm \ref{alg:iterative_approx_of_uniformization_map} is that once a feasible solution is found $\Phi^q\in\FF^{\S^q}_{\K,\Tau}$. It is possible to further improve the result by pushing the conformal distortion per face further down using the following, now feasible, problem:
\begin{eqnarray}\nonumber
\min &\sum_{j}\eps_j& \\ \nonumber
\mathrm{s.t.} & &\\ \label{e:convex_feasibility_given_Tau_improve}
\eps_j & \leq & \eps \\ \nonumber
|\beta_j| & \leq & \frac{K_j-1}{K_j+1}\re \parr{e^{-\i \tau_j} \alpha_j} +\eps_j,\\ \nonumber
\textrm{and} && eqs.~(\ref{e:compatibility_eqs}),(\ref{e:corner_vertices_to_triangle_corners}),(\ref{e:boundary_edges_1}),(\ref{e:boundary_edges_2}).
\end{eqnarray}
These equations take the role of eqs.~(\ref{e:convex_feasibility_given_Tau}). We now iterating with these equations, keep updating $\Tau$ and resolving until $\sum_j \eps_j$ is not decreasing.

Another comment is that in implementing Algorithm \ref{alg:iterative_approx_of_uniformization_map} it is possible to replace the cone conditions in (\ref{e:convex_feasibility_given_Tau}) (and similarly in (\ref{e:convex_feasibility_given_Tau_improve})) with linear constraints to end up with standard linear programming problem:

\begin{eqnarray}\nonumber
\min &\eps& \\ \label{e:convex_feasibility_given_Tau_linprog}
\mathrm{s.t.} & &\\ \nonumber
|\beta_j|_\infty & \leq & \frac{1}{\sqrt{2}}\frac{K_j-1}{K_j+1}\re \parr{e^{-\i \tau_j} \alpha_j} +\eps,\\ \nonumber
\textrm{and} && eqs.~(\ref{e:compatibility_eqs}),(\ref{e:corner_vertices_to_triangle_corners}),(\ref{e:boundary_edges_1}),(\ref{e:boundary_edges_2}),
\end{eqnarray}
where we denote $|\beta|_\infty=\max\set{\abs{\re\beta},\abs{\im\beta}}$. It is shown in \cite{Lipman:2012:BDM:2185520.2185604} (see Proposition 4.2), that this formulation will be equivalent to the cone condition as $K_j\too 1$.

Figures \ref{fig:cat_head} and \ref{fig:igea} show two examples of approximations of the uniformization maps calculated with our iterative algorithm. The first row shows the subdivided series of meshes $\S^0=\S \prec \S^1 \prec \S^2 \prec \S^3$. The second (and third row in Figure \ref{fig:igea}) shows a checkerboard texture mapped to the polyhedral surface by the inverse of the simplicial approximations $\Phi^q$. The bottom row shows the homeomorphic image under $\Phi^q$ of the triangulation $\S^q$ onto the equilateral. Red color indicates conformal distortion. Note how the approximation improve as the refinement level increases.

\begin{figure}[t] %\vspace{-0.5cm}
\centering
\begin{tabular}{@{\hspace{0.0cm}}c@{\hspace{0.05cm}}c@{\hspace{0.0cm}}c@{\hspace{0.0cm}}c@{\hspace{0.0cm}}}
  \includegraphics[width=0.25\columnwidth]{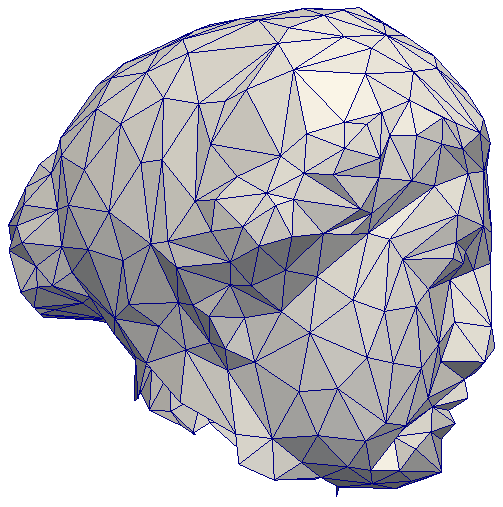}&
    \includegraphics[width=0.25\columnwidth]{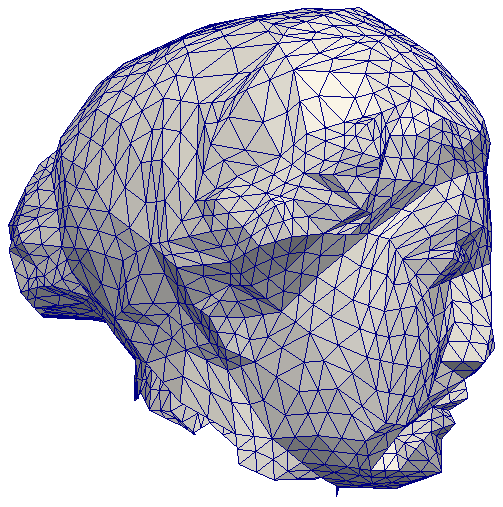}&
      \includegraphics[width=0.25\columnwidth]{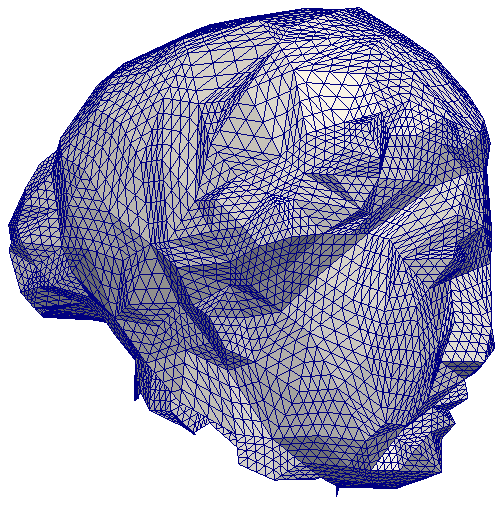}&
        \includegraphics[width=0.25\columnwidth]{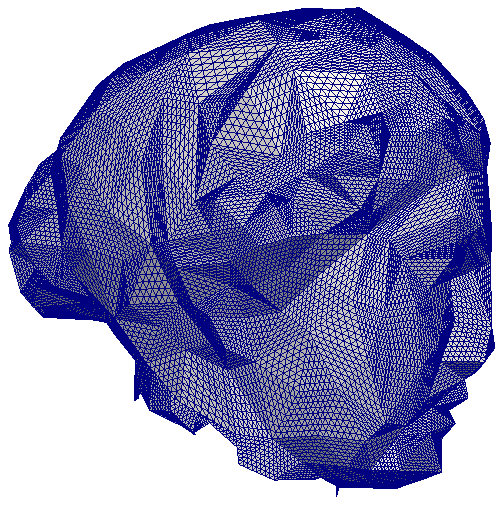}\\
&
    \includegraphics[width=0.25\columnwidth]{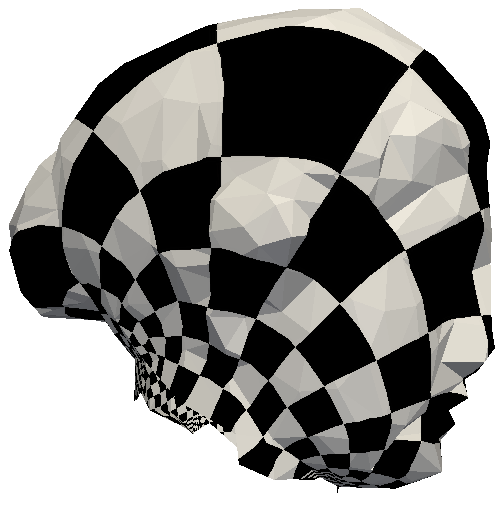}&
      \includegraphics[width=0.25\columnwidth]{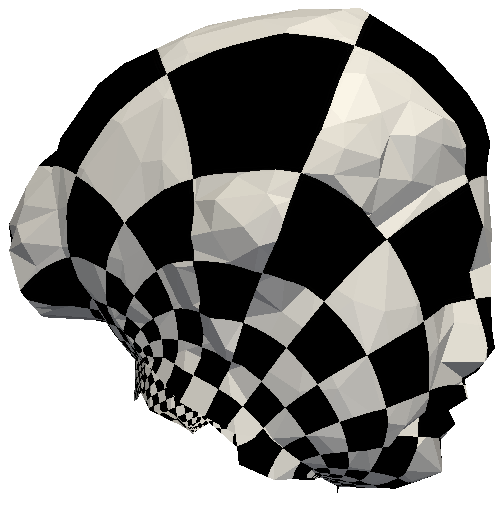}&
        \includegraphics[width=0.25\columnwidth]{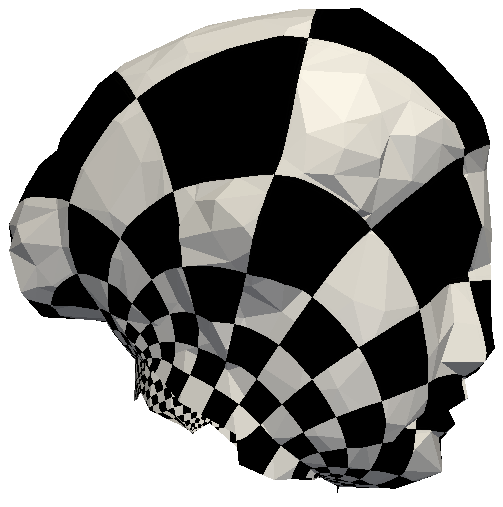}\\

&
    \includegraphics[width=0.25\columnwidth]{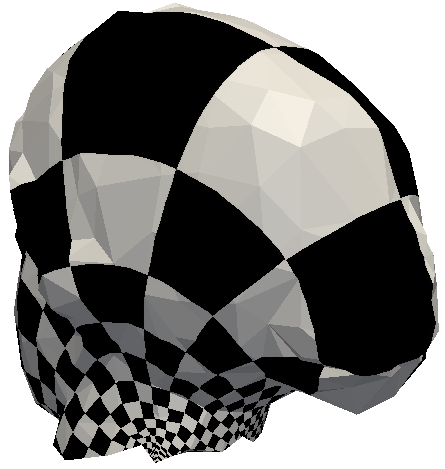}&
      \includegraphics[width=0.25\columnwidth]{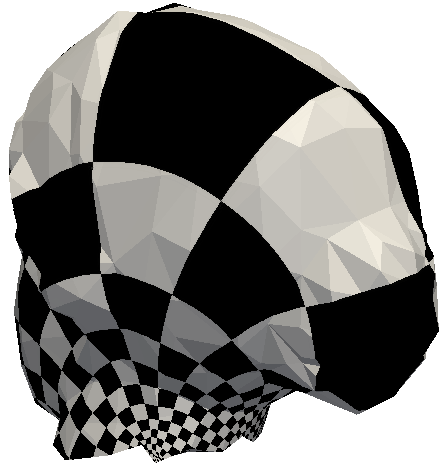}&
        \includegraphics[width=0.25\columnwidth]{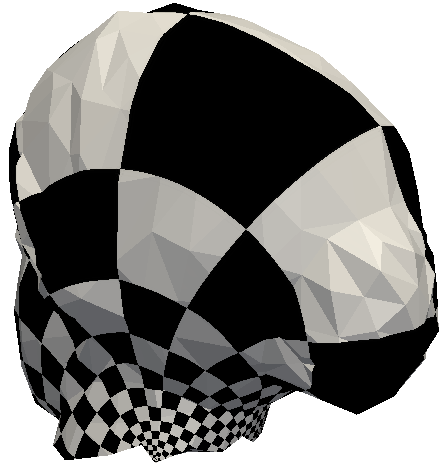}\\

&
    \includegraphics[width=0.25\columnwidth]{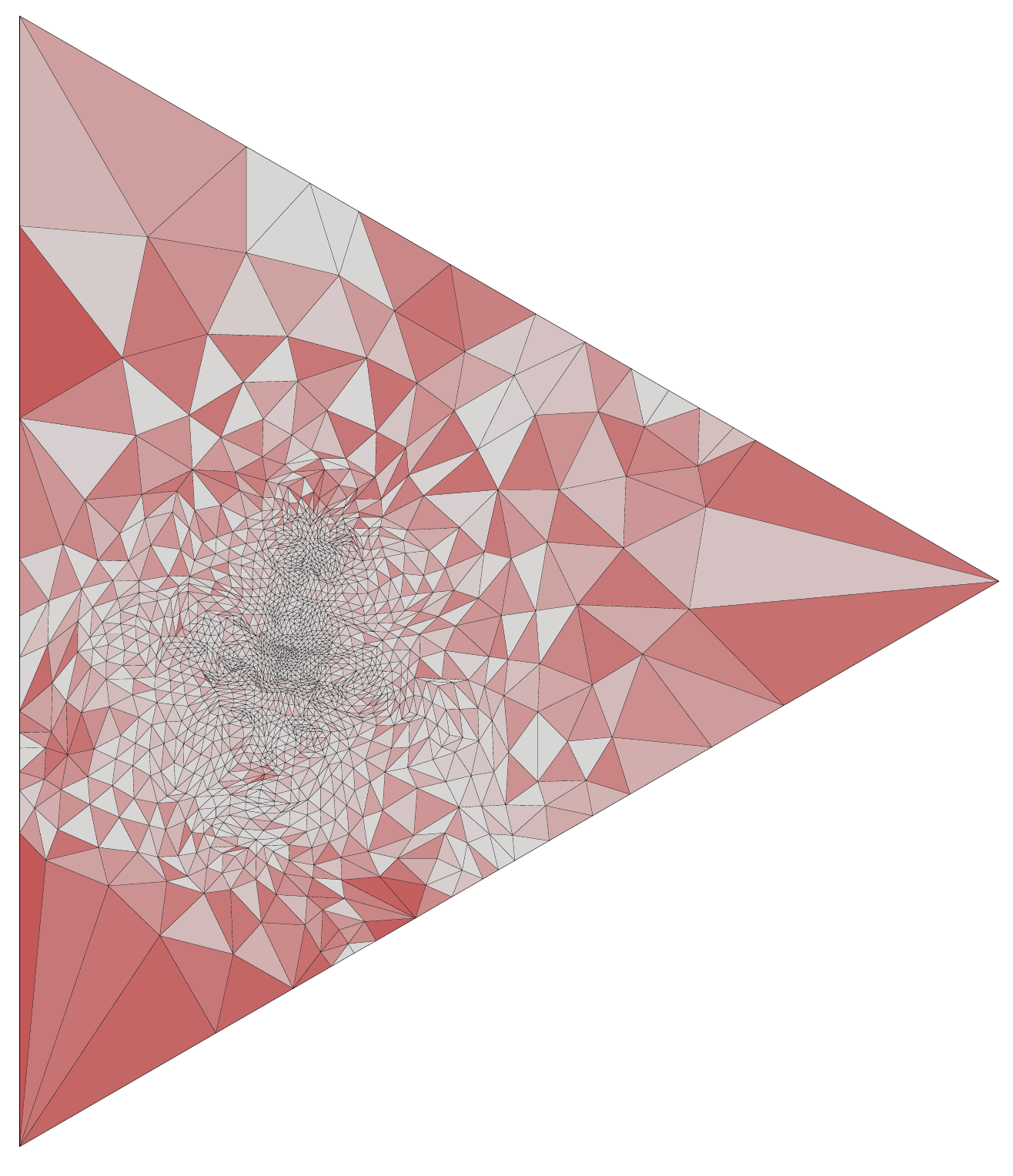}&
      \includegraphics[width=0.25\columnwidth]{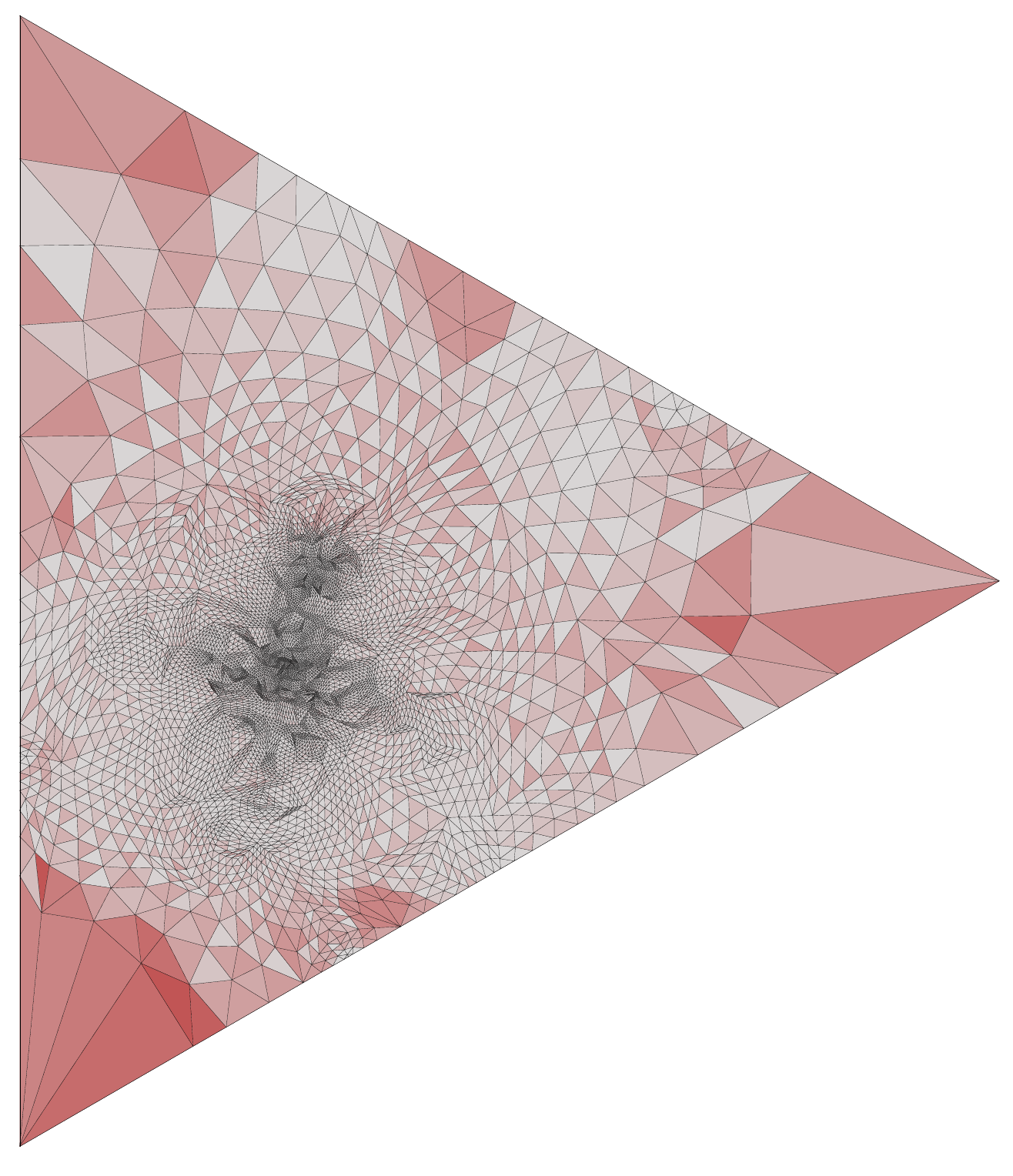}&
        \includegraphics[width=0.25\columnwidth]{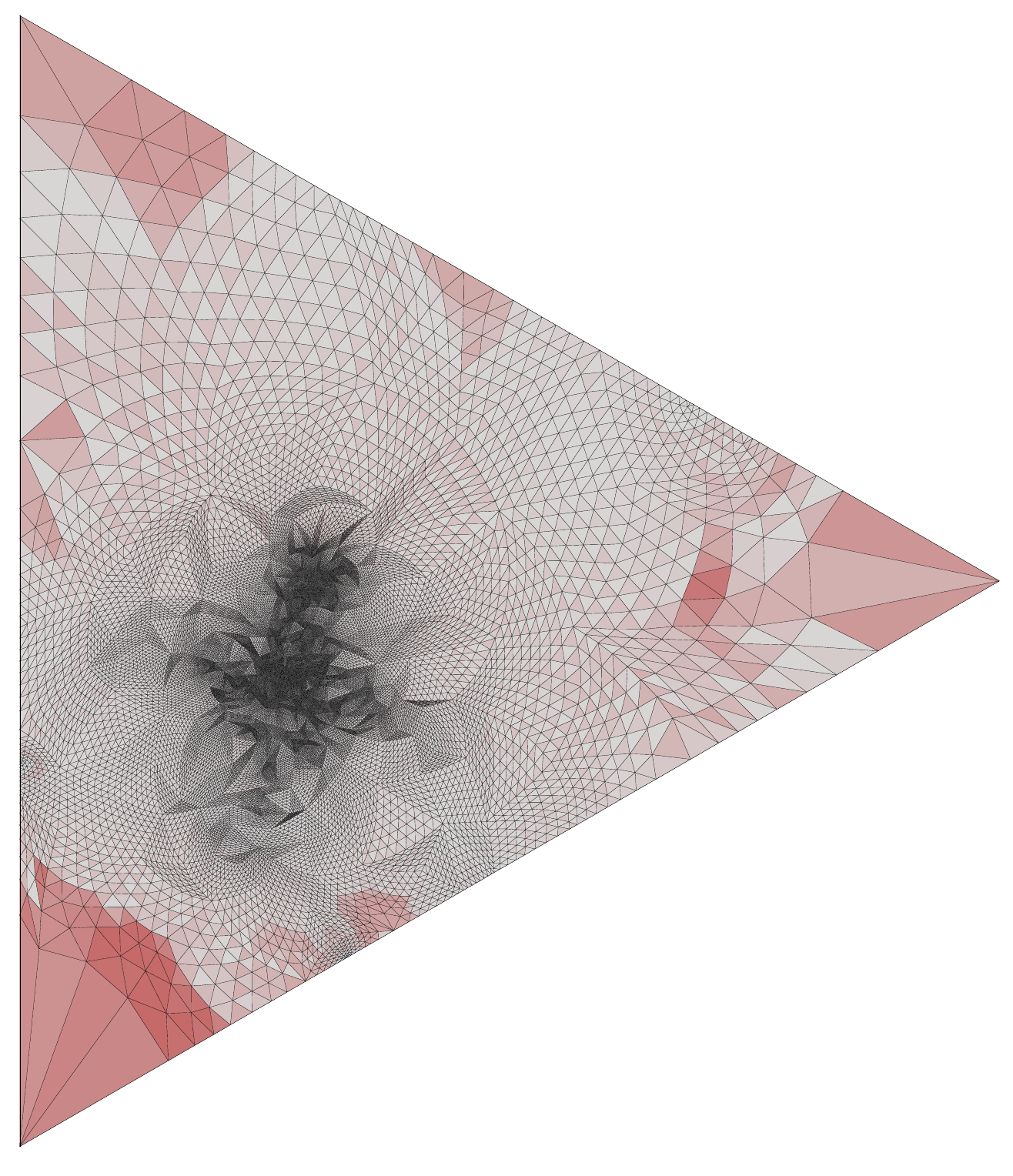}\\
$\S$ & $\S^1$ & $\S^2$ & $\S^3$ \\
\end{tabular}
  \caption{Approximation of the uniformization map of the Igea polyhedral surface. }\label{fig:igea}
\end{figure}

%
%\textbf{Practical implementation and concluding remarks}
%We know $\wh{\tau}$ on the boundary and the arg is harmonic...
%
%\begin{enumerate}
%\item
%Set $\eps>0$ in (\ref{e:alpha_j_leq_dil_beta_j}).
%\item
%Add some constant to algorithm $K_j=1+C 2^{-c q \gamma_{i_j}}$.
%\item
%Use only $u_i$ as variables.
%\item describe the convex problem (cone etc) and also relaxation ($L_2$). Show the relaxation does not hurt us.
%\end{enumerate}

\newpage

%%%%%%%%%%%%%%%%%%%%%%%%%%%%%%%%%%%%%%%%%%%%%%%%%%%%%%%%%%%%%%%%%%%%%%%%%%%%%%%%%%%%%%%%%
\appendix
\renewcommand{\thesection}{\Alph{section}}
\section{}
\label{a:auxilary_lemmas}
\renewcommand{\thesection}{\Alph{section}}

\textbf{Lemma \ref{lem:discretization_of_power_map}}
\textit{
Let $T=\Delta(\xi,0,\eta)$ be an isosceles triangle ($|\xi|=|\eta|$) and denote the angle $\theta=\measuredangle(\xi,0,\eta)$. Further let $T^q$ be the $q^{th}$ level of regular 1-4 subdivision of $T$. Denote by $h(z)=z^\gamma$ the power map, and assume that $\lceil \gamma \rceil \theta < \frac{\pi}{2}$, and that $\theta < 60.4^\circ$. Then, the simplicial maps $h^q$ defined by sampling $h(z)$ over the vertices of $T^q$ and extending by linearity are homeomorphisms that satisfy $\CD(h^q)\leq K$ for some $K\geq 1$ independent of $q$..}
\begin{proof} We will denote by $\chi^{q}:T\too\C$ the simplicial map that is defined by sampling $h(z)=z^\gamma$ at the vertices of $T^q$ and extending linearly. In particular, $\chi^q$ will map affinely the triangle $\Delta=\Delta(z_1,z_2,z_3)\subset T^q$ to the triangle $\Delta(z_1^\gamma, z_2^\gamma, z_3^\gamma)$. We need to prove that $\chi^{q}$ are all quasiconformal maps with a universal bound on their conformal distortion (a bound independent of $q$).

We first prove that it is enough to show that for sufficiently large $Q>0$, $\chi^q$, $q\geq Q$ is a homeomorphism. Then we show that such a $Q$ indeed exists.

Let us denote the set $T_{far}\subset T$ to include all the points $z\in T$ such that if we write $z=\lambda_1 \xi + \lambda_2 \eta$ then $\lambda_1+\lambda_2\geq \frac{1}{2}$. We also denote $T_{near}=\closure{T\setminus T_{far}}$. Lemma \ref{lem:sampling_conformal_with_triplet} indicates that there exists some finite level of subdivision $\wt{Q}$, such that for all subdivision levels $q\geq \wt{Q}$, the map $\chi^q$ restricted to the set of triangles in $T^q$ contained in $T_{far}$ is quasiconformal with bounded conformal distortion of, say, $\wt{K}=2$. That is, $\CD(\chi^q\mid_{T_{far}})\leq 2$ for $q\geq \wt{Q}$. Now our assumption is that for sufficiently large $Q$, $\chi^q$ is homoemorphic for $q\geq Q$. In particular we can take $Q\geq\wt{Q}$.

Now, $\chi^Q$ is homeomorphic, and therefore has to be quasiconformal since it maps finite simplicial complex (build out of finitely many triangles). Denote $\CD(\chi^Q)=K_Q$. Since $h(az)=(az)^\gamma=a^\gamma h(z)$, and multiplying by a complex number keeps  distortion and orientation unchanged we see that the conformal distortion of $\chi^{Q+1}$ over triangles in $T^{Q+1}$ contained in  $T_{near}$ is also bounded by $K_{Q}$. Therefore, the conformal distortion $\CD(\chi^{Q+1}\mid_{T_{near}})\leq K_{Q}$, while we already know that $\CD(\chi^{Q+1}\mid_{T_{far}})\leq 2$. Therefore $$\CD(\chi^{Q+1}\mid_{T_{near}})\leq \max\set{K_{Q},2}.$$
Continuing with induction we prove that $$\CD(\chi^{q}\mid_{T_{near}})\leq \max\set{K_{Q},2},$$ for all $q\geq Q$.

\begin{floatingfigure}[l]{0.3\textwidth}
%\begin{wrapfigure}{l}{0.3\textwidth}
  \begin{center}
    \includegraphics[width=0.3\textwidth]{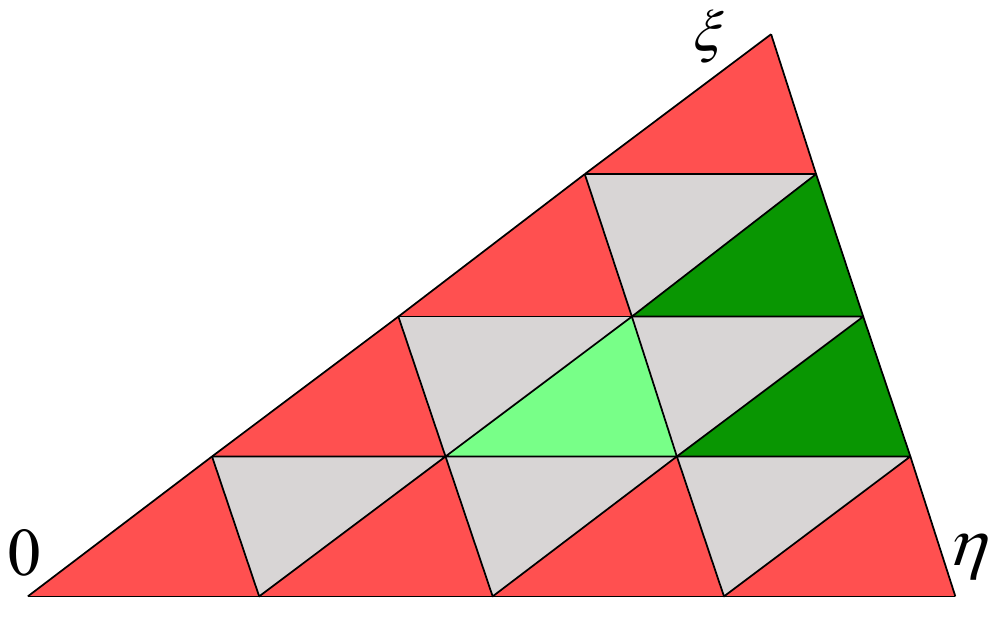}%\vspace{-0.2cm}
  \end{center}
  %\caption{The Projection approach.}
  %\label{fig:multi_connected}
%\end{wrapfigure}
\end{floatingfigure}
We are left with the task of proving that $\chi^q$, $q\geq Q$ are all homeomorphisms for sufficiently large $Q$. We will refer to the inset image depicting an example of $\Delta(\xi,0,\eta)$ for illustration. It is enough to show that every triangle in every subdivision level $T^q$ is not flipped by $\chi^q$. Indeed since the boundary polygonal of $T^q$ is mapped to another simple polygon, the fact that all inner triangles maintain their orientation implies homeomorphism (see, e.g., \cite{Lipman:2012:BDM:2185520.2185604} for a proof).

Let us fix some arbitrary $q$ (e.g., see the inset). The first observation is that the ``corner'' triangle, namely the triangle touching the origin does not flip orientation under $\chi^q$. This can be seen by using the fact that $\lceil\gamma\rceil\theta<\frac{\pi}{2}$. This actually implies that all triangles with an edge on the segments $[0,\xi]$ and $[0,\eta]$ (colored red in inset) are not flipped.\\
We are left with the triangles of the type colored green and gray in the inset. Since, as we noted before, scaling does not change the orientation we can consider w.l.o.g. triangles of the form $\Delta(1,1+\eta,1+\xi)\in T^q$. Since every triangle in $T^q$ is similar to $T$, we can further assume that the corner at $1$ corresponds to $0$ in $T$.\\
Since $\theta < 60.4^{\circ} < \pi/2$ one must have $\re(\xi)\re(\eta)>0$. The triangles for which $\re(\xi)<0$ and $\re(\eta)<0$ (corresponds to the gray triangles) are not flipped since $h(1)=1$, $h(1+\xi)=(1+\xi)^\gamma$, $\abs{(1+\xi)^\gamma}<1$, and $\arg\parr{(1+\xi)^\gamma}<\frac{\pi}{2}$ since $\theta\lceil\gamma\rceil<\frac{\pi}{2}$ and $\arg(1+\xi)\leq\theta$. Similarly, $\arg\parr{(1+\eta)^\gamma}>-\frac{\pi}{2}$ and therefore the triangle $\Delta(1,h(1+\eta),h(1+\xi))$ has the same orientation as $\Delta(1,1+\eta,1+\xi)$. \\
We are left with the triangles for which $\re(\xi)>0$ and $\re(\eta)>0$ (the green triangles in the inset). We will later use the fact that since all the triangles are isoceles a consequence from the cosine law is that,
for all ``green'' triangles,\begin{equation}\label{e:bound_r}
|\xi|\leq  \frac{1}{\sqrt{2(1+\cos(\theta))}},
\end{equation} and similar bound holds for $|\eta|$.\\
A sufficient condition that will prevent triangle $\Delta(1,1+\eta,1+\xi)$ from flipping is $$\re(1+\xi)^\gamma-1 > 0\ , \ \re(1+\eta)^\gamma -1>0.$$ We will show the first inequality, and the second is proved similarly. In the following we use the binomial expansion \cite{ahlfors1979complex}.  Denote $\xi=r e^{\i \vphi}$, then
\begin{eqnarray}\nonumber
\re(1+\xi)^\gamma-1 &=& \gamma~\re\xi+\frac{\gamma(\gamma-1)}{2}~\re\xi^2+\frac{\gamma(\gamma-1)(\gamma-2)}{3!}~\re\xi^3+...\\\nonumber%\frac{\gamma(\gamma-1)(\gamma-2)(\gamma-3)}{4!}~\re\xi^4+...\\
&=& \gamma r\cos(\vphi)+\frac{\gamma(\gamma-1)}{2}r^2\cos(2\vphi)+\frac{\gamma(\gamma-1)(\gamma-2)}{3!}r^3\cos(3\vphi)+...\\\nonumber%\frac{\gamma(\gamma-1)(\gamma-2)(\gamma-3)}{4!}r^4\cos(4\vphi)+...
&=& \gamma r\parr{\cos(\vphi)+\frac{\gamma-1}{2}r\cos(2\vphi)+\frac{(\gamma-1)(\gamma-2)}{3!}r^2\cos(3\vphi)+...}\\ \label{e:re(1+xi)^gamma-1}
&=& \gamma r\parr{\sum_{j=0}^\infty a_j\cos\parr{(j+1)\vphi}r^j},
\end{eqnarray}
where $a_j = \frac{(\gamma-1)(\gamma-2)\cdot\cdot\cdot(\gamma-j)}{2\cdot3\cdot...(j+1) }$, for $j=1,2,...$, and $a_0=1$. Let $k\in \Natural$ be such that $k<\gamma\leq k+1$ . Since $0<\gamma-k\leq1$ we have
$$\abs{a_k} = \abs{\frac{(\gamma-1)(\gamma-2)\cdot...\cdot(\gamma-k)}{2\cdot3\cdot...\cdot(k+1) }}\leq \frac{1}{k+1}.$$
Using induction one proves that for $j=k+1,k+2,...$
$$\abs{a_j}\leq \frac{1}{j+1}.$$
We can now bound (\ref{e:re(1+xi)^gamma-1}) as follows
  \begin{eqnarray}\nonumber
\re(1+\xi)^\gamma-1 &=& \gamma r\parr{\sum_{j=0}^k a_j\cos\parr{(j+1)\vphi}r^j\ +\ \sum_{j=k+1}^\infty a_j\cos\parr{(j+1)\vphi}r^j}\\\nonumber
&\geq& \gamma r\parr{\cos(\vphi) - \sum_{j=k+1}^\infty \frac{1}{j+1}r^j}\\\nonumber
&\geq& \gamma r\parr{\cos(\vphi) - \frac{1}{r}\sum_{j=k+1}^\infty \frac{1}{j+1}r^{j+1}}\\\nonumber
&\geq& \gamma r\parr{\cos(\vphi) - \frac{1}{r}\parr{-\log\parr{1-r}-r}}\\\label{e:re(1+xi)^gamma_final_bound}
&\geq& \gamma r\parr{\cos(\theta) + \frac{\log\parr{1-r}}{r}+1},
\end{eqnarray}
where in the second row we used the fact that for $1\leq j \leq k$, $j+1\leq \lceil\gamma\rceil$, and therefore $\cos((j+1)\vphi)\geq\cos(\lceil\gamma\rceil\theta)\geq 0$ as $\theta \lceil\gamma\rceil<\pi/2$.
Plugging (\ref{e:bound_r}) in (\ref{e:re(1+xi)^gamma_final_bound}) we get
\begin{eqnarray*}
\re(1+\xi)^\gamma-1 &\geq& \gamma r \parr{\cos(\theta)+\frac{\log\parr{1-\brac{2(1+\cos(\theta))}^{-1/2}}}{\brac{2(1+\cos(\theta))}^{-1/2}}+1}\\ &=& \gamma r \Upsilon(\theta),
\end{eqnarray*}

\begin{floatingfigure}[r]{0.4\textwidth}
%\begin{wrapfigure}{l}{0.3\textwidth}
  \begin{center}\vspace{-1.5cm}
    \includegraphics[width=0.2\textwidth]{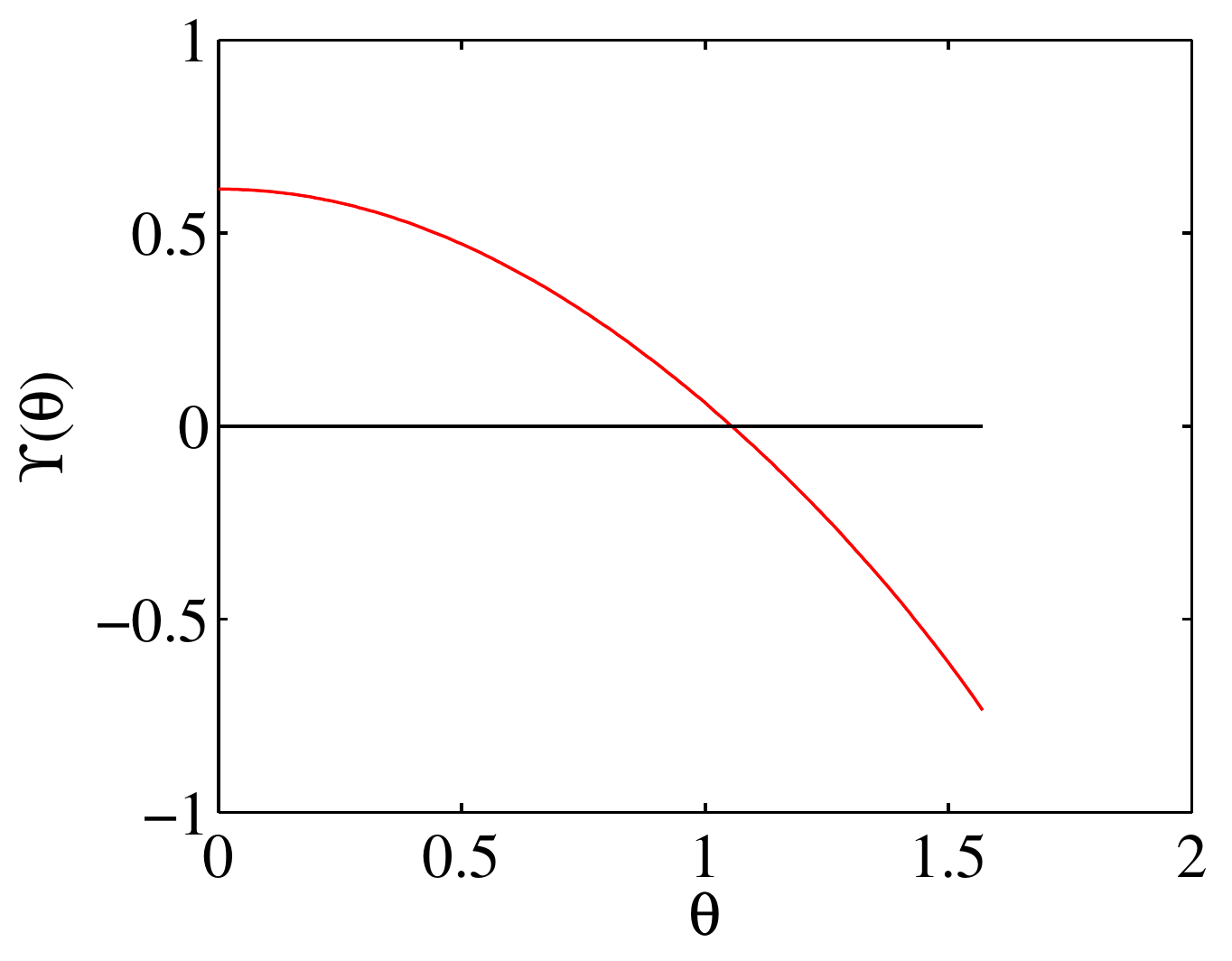}%\vspace{1cm}
  \end{center}
  %\caption{The Projection approach.}
  %\label{fig:multi_connected}
%\end{wrapfigure}
\end{floatingfigure}

Investigating the function $\Upsilon(\theta)$ (see inset for its graph, in red) reveals that for $\theta\leq 60.4^\circ$ $\Upsilon(\theta)>0$ and therefore $\re(1+\xi)^\gamma-1>0$. Similarly, $\re(1+\eta)^\gamma-1>0$ and therefore the triangle $\Delta(1,h(1+\eta),h(1+\xi))$ is not flipped.
\end{proof}

\bibliographystyle{amsplain}
\bibliography{discrete_qc}
\end{document}